\documentclass{article}
\usepackage{graphicx} 
\usepackage[utf8]{inputenc}
\usepackage{upref}
\usepackage{amsmath}
\usepackage{amsthm}
\usepackage{comment}
\usepackage{amsfonts}
\usepackage{dsfont}
\usepackage{amssymb}
\usepackage{geometry}
\usepackage{esint}
\usepackage{bbm}
\usepackage[table]{xcolor}
\usepackage{graphicx}
\usepackage{tcolorbox}
\usepackage{color}
\usepackage{cancel}
\usepackage[colorlinks=true, linkcolor=black, citecolor=red] {hyperref}
\usepackage[autostyle=true]{csquotes}
\usepackage{authblk}
\usepackage{appendix}

\newtheorem*{lemma*}{Lemma}

\newtheorem{theorem}{Theorem}[section]
\newtheorem{lemma}[theorem]{Lemma}

\newtheorem{proposition}[theorem]{Proposition}

\theoremstyle{definition}
\newtheorem{definition}[theorem]{Definition}

\theoremstyle{remark}
\newtheorem{remark}[theorem]{Remark}

\theoremstyle{definition}
\newtheorem{notation}[theorem]{Notation}

\DeclareMathOperator{\R}{\mathbb{R}}
\DeclareMathOperator{\C}{\mathbb{C}}
\DeclareMathOperator{\N}{\mathbb{N}}
\DeclareMathOperator{\U}{\mathbb{U}}
\DeclareMathOperator{\Ham}{\mathbb{H}}
\DeclareMathOperator{\Hcal}{\mathcal{H}}
\DeclareMathOperator{\Vcal}{\mathcal{V}}
\DeclareMathOperator{\Acal}{\mathcal{A}}
\DeclareMathOperator{\Fcal}{\mathcal{F}}

\DeclareMathOperator{\Ncal}{\mathcal{N}}
\DeclareMathOperator{\Gcal}{\mathcal{G}}

\DeclareMathOperator{\Scal}{\mathcal{S}}

\DeclareMathOperator{\dg}{d\Gamma}
\DeclareMathOperator{\tr}{Tr}
\DeclareMathOperator{\hl}{\mathit{h}_{\lambda}}
\DeclareMathOperator{\Hl}{\mathbb{H}_{\lambda}}
\DeclareMathOperator{\Ul}{\mathbb{U}_{\lambda}}
\DeclareMathOperator{\xil}{\xi_{\lambda}}
\DeclareMathOperator{\Span}{\mathrm{span}}

\renewcommand{\d}{ \, \mathrm{d}}

\numberwithin{equation}{section}

\begin{document}

\title{Renormalization of bosonic quadratic Hamiltonians involving rank-one perturbations}
\author{Thomas Gamet\thanks{Ecole Normale Supérieure de Lyon, UMPA (UMR 5669), Lyon, France}~\thanks{thomas.gamet@ens-lyon.fr}}
\date{}
\maketitle
\begin{abstract}
    We study the renormalization of a bosonic quadratic Hamiltonian with an ultraviolet divergence. The Hamiltonian is composed of the sum of a free part and the square of the smeared field operator. We explicitly diagonalize the Hamiltonian via Bogoliubov transformations, thus simplifying its definition as a self-adjoint operator. Depending on the field operator's smearing, we discuss different renormalizations, either of the energy alone, or the energy and coupling constant together. 
\end{abstract}

\tableofcontents

\section{Introduction}

\subsection{General overview}

We study a particular class of formal quadratic bosonic Hamiltonians depending on a parameter -- the regularity of the interaction -- which are the sum of a free Hamiltonian and the square of a smeared field operator. In Quantum Field Theories (QFT), quadratic Hamiltonians are used to describe a free bosonic field interacting with an external field; N. Bogoliubov notably used them in his theory of interacting Bose gas \cite{bogoliubov1947theory}. The study of these models often leads to divergence problems. In practice, this means that if we naively compute the energy or scattering amplitude, these physical quantities may appear to be infinite. Mathematically, the difficulty is to determine a self-adjoint Hamiltonian, along with its domain. In our specific case, we deal with an ultraviolet divergence, i.e. the divergence comes from the contributions of local singularities. In order to compute physically relevant quantities, the theory must be renormalized, that is divergences must be compensated. Here, we renormalize both the energy and the coupling constant (or charge). More specifically, we want to study the formal normal-ordered Hamiltonian given by
\begin{equation}\label{equation_normal_formal_hamiltonian}
    \dg(\omega) + \lambda {:}\big(a^*(f) + a(f)\big)^2{:}
\end{equation}
where $a^*(f)$ and $a(f)$ are respectively the creation and annihilation operators, and $\dg(\omega)$ the second quantification of $\omega$. To study this Hamiltonian, we use Bogoliubov transformations, which are unitary transformations that preserve the Canonical Commutation Relations (CCR). In some cases, quadratic Hamiltonians can be diagonalized by Bogoliubov transformations, and the diagonalized Hamiltonian is that of free (i.e. non interacting) quasi-particles. In our case, the model is quite simple, which allows us to explicitely compute the diagonalized Hamiltonian. Depending on the regularity of $f$, we show that there are four cases:
\begin{itemize}
    \item the formal Hamiltonian (\ref{equation_normal_formal_hamiltonian}) directly defines a self-adjoint operator;
    \item the formal Hamiltonian (\ref{equation_normal_formal_hamiltonian}) does not directly define a self-adjoint operator, but the diagonalized one does without any renormalization;
    \item the formal Hamiltonian (\ref{equation_normal_formal_hamiltonian}) does not directly define a self-adjoint operator, but the diagonalized one does after a renormalization of the energy;
    \item the formal Hamiltonian (\ref{equation_normal_formal_hamiltonian}) does not directly define a self-adjoint operator, but the diagonalized one does after a renormalization of both the energy and the coupling constant $\lambda$.
\end{itemize}
In the latter case, we will show that there are no Bogoliubov transformations between the formal Hamiltonian and the renormalized one, meaning  that the models are not equivalent. Our model is quite simple, therefore we can understand difficulties of the renormalization that may be general. In particular, in this model a renormalization of the charge occurs.~\\

Let us now compare our model and results to the literature. First, note that our model is  a special case of the general quadratic Hamiltonian diagonalized in \cite{nam_diagonalization_2016}. The ultraviolet renormalization of quadratic Hamiltonians has been notably studied for $\omega(k) = \sqrt{k^2 + m^2}$, that is, for the kinetic energy of a massive relativistic particle in \cite{guenin_mass_1968,ginibre_renormalization_1970}. In both papers, the authors show that the addition of a specific quadratic perturbation to the free energy $\dg(\omega)$ corresponds to a renormalization of the mass. In particular, Bogoliubov transformations are used in \cite{ginibre_renormalization_1970} to diagonalize the Hamiltonian. More recently and in a more general setting, the definition and self-adjointness of bosonic quadratic Hamiltonians have been studied in \cite{bruneau_bogoliubov_2007, derezinski_bosonic_2017}. In these papers, the Hamiltonians are defined as generators of Bogoliubov transformations, and correspond to different quantizations of the same classical quadratic Hamiltonian. A renormalization of the energy -- but not the charge -- is described. The Van Hove Hamiltonian \cite{van_hove_difficultes_1952} is the linear counterpart of our model. It is composed of the free energy, i.e. the second quantization of the one-body operator and of a term which is linear in the creation and annihilation operators. It can be written as
\begin{equation}
    \mathbb{H}_{\mathrm{VH}} = \dg(\omega) + a(f) + a^*(f).
\end{equation}
The renormalization of this Hamiltonian has been studied in \cite{derezinski_van_2003}. In this case, a dressing operator is used to diagonalize the Hamiltonian:
\begin{equation}\label{equation_van_hove_renormalized}
    \dg(\omega) + a(f) + a^*(f) + \|\omega^{-1/2}f\|^2 = U\dg(\omega)U^*.
\end{equation}
As long as $\omega^{-1/2}f$ is well defined, we can make sense of the Hamiltonian. Otherwise, $\mathbb{H}_{\mathrm{VH}}$ is not defined, but as long as $\omega^{-1}f$ is well defined, $U$ exists and the renormalized Hamiltonian in (\ref{equation_van_hove_renormalized}) is well defined thanks to the diagonalized expression. In the case above, there is a renormalization of the energy of a Hamiltonian with a linear perturbation, while in ours, there is a renormalization of both the energy and the charge of a Hamiltonian with a quadratic perturbation. A detailed study of the domain of the Hamiltonian in a special case may be found in \cite{lampart_particle_2018}. The Van Hove Hamiltonian, as well as our model, are simpler models that allow to study separately different parts of the Pauli Fierz Hamiltonian, which describes the interaction between an electron and a field of photons with minimal coupling \cite{pauli_zur_1938}. It can be written as
\begin{equation}
    \mathbb{H}_{\mathrm{PF}} = \dg(\omega) + V_{f, \mathrm{Coulomb}} + (p - \lambda A_f(x))^2,
\end{equation}
where $A_f$ is a field operator with form factor $f$. The ultraviolet renormalization of the Pauli-Fierz Hamiltonian is not fully understood yet (see Section 19.3 of \cite{spohn_dynamics_2004} for a perturbative study), and simpler models may help understand it. When expanding $(p - \lambda A_f(x))^2$, a linear and a quadratic term in the field operator appear. Keeping only the linear term corresponds to the Van Hove Hamiltonian, and our Hamiltonian (\ref{equation_normal_formal_hamiltonian}) corresponds to the quadratic part. The Pauli-Fierz Hamiltonian is one of several polaron models, whose study is an active field of research, see for instance \cite{schmidt_massless_2021, dam_absence_2022, alvarez_ultraviolet_2023, falconi_bogoliubov_2023,  lill_self-adjointness_2023, hinrichs_ultraviolet_2025}.\\

We mainly use three different methods inspired by the literature. First, we diagonalize the Hamiltonian in the more regular case, which is a particular case of the diagonalization of a general quadratic Hamiltonian realized in \cite{nam_diagonalization_2016} (previous works on the subject include \cite{friedrichs_mathematical_1953, lewin_bogoliubov_2015, kato_friedrichs-berezin_1967}). Due to the special form of our Hamiltonian, following ideas of \cite{grech_excitation_2013}, we are able to give a more explicit diagonalization of the Hamiltonian. Then, to renormalize the Hamiltonian, we use the renormalization of the one-particle operator, which can be found in \cite{kiselev_rank_1995, derezinski_renormalization_2002}. Finally, we use integral formulas several times to study roots of operators. These formulas are the generalization of an equality found in \cite{ravn_christiansen_random_2021}, following \cite{benedikter_optimal_2020}.

\subsection{Mathematical setting and definition of the Hamiltonian}
\begin{notation}
Let $\Hcal$ be a  complex separable Hilbert space. We denote by $\langle\cdot|\cdot\rangle$ the Hermitian product of $\Hcal$, which is linear in its second variable. We denote by $\Hcal^*$ the topological dual space of $\Hcal$, that is the set of continuous linear forms on $\Hcal$, and we set $J : |\psi\rangle \mapsto \langle \psi|$ the canonical isometry between $\Hcal$ and $\Hcal^*$. Note that $J$ is antilinear and satisfies $J^* = J^{-1}$. If $\omega$ is a self-adjoint operator on $\Hcal$, we write $\mathds{1}_{\omega \leq \alpha}$ the orthogonal projector on the space on which $\omega \leq \alpha$. Moreover, we denote the resolvent set of $\omega$ by $\rho(\omega)$, and the resolvent by $R_z(\omega) = (\omega - z)^{-1}$ for $z\in \rho(\omega)$. We want to consider vectors that do not lie in $\Hcal$.  We assume $\omega \geq1$ and $D(\omega) \varsubsetneq \Hcal$, where $D(\omega)$ is the domain of $\omega$. Let $\|\psi\|_s = \|\omega^{-s}\psi\|$ and $D(\omega^{-s})$ be the completion of $\big(\Hcal, \|\cdot\|_s\big)$, for all $\psi \in D(\omega^{-s})$, we have $\omega^{-s} \psi \in \Hcal$. In the following, we will write equivalently $\omega^{-s}\psi \in \Hcal$ or $\psi \in D(\omega^{-s})$. 

For instance, if $\Hcal = L^2(\R)$, $\omega$ is multiplication operator such that $\omega = \sqrt{1+|k|^2}$ and $f : k\mapsto |k|^{\alpha}$ with $\alpha > -1/2$, we have $\omega^{-s} f \in \Hcal$ if and only if $2\alpha - 2s < - 1$. More specifically, if $f(k) = 1$, which corresponds to the Dirac distribution in the Fourier representation, we have $\omega^{-s}f \in \Hcal$ (that is $(\sqrt{1-\Delta})^{-s} \delta \in L^2$) if and only if $s > 1/2$. For a more detailed presentation of scales of Hilbert spaces generated by a self-adjoint operator, see for instance Section 1.2.2 of \cite{albeverio_singular_2000}.
\end{notation}
We introduce the bosonic Fock space $\Fcal$, allowing us to consider states with a variable number of particles:
\begin{equation}
    \Fcal = \bigoplus_{n \in \N} \Hcal^{\otimes_s n}
\end{equation}
where $\Hcal^{\otimes_s n}$ is the $n$-particles space with the convention $\Hcal^{\otimes_s 0} = \C$. Every vector $\Psi$ of $\Fcal$ can be written as $\Psi = (\psi_n)_{n\in \N}$ with $\psi_n \in \Hcal^{\otimes_s n}$. We define the vacuum $\Omega$ by $\Omega = (\delta_{n,0})_{n \in \N}$. If $\omega$ is a positive operator on $\Hcal$ with domain $D(\omega)$, we set 
\begin{equation}
    \d\Gamma(\omega) = \bigoplus_{n \in \N} \left(\sum_{i=1}^n \omega_i\right)= 0 \oplus \omega \oplus (\omega \otimes 1 + 1 \otimes \omega) \oplus ...
\end{equation}
on $\bigcup_{N\in \mathbb{N}}\bigoplus_{n = 1}^N D(\omega)^{\otimes_s n}$, and we call $\d\Gamma(\omega)$ the second-quantization of $\omega$. The operator $\dg(\omega)$ is essentially self-adjoint and we still denote by $\d \Gamma(\omega)$ its Friedrichs' extension. If $\omega$ is the Hamiltonian of a particle, then $\d\Gamma(\omega)$ is the Hamiltonian of a system of free particles, with a variable number of particles. In particular, the number operator is $\Ncal = \d\Gamma(1)$. For a unitary operator $U$ on $\Hcal$, we set
\begin{equation}
    \Gamma(U) = \bigoplus_{n\in \N}\bigotimes_{i=1}^n U_i = 1 \oplus U\oplus(U\otimes U) \oplus...
\end{equation}
which is unitary on $\Fcal$. Let $f\in \Hcal$ and $u_1,...,u_n  \in \Hcal$, we define respectively the creation and annihilation operators by\begin{equation}
    \begin{cases}a^*(f) \big(u_1 \otimes_s...\otimes_s u_n\big) = \sqrt{n+1}\big(f\otimes_s \otimes u_1...\otimes_s u_n\big)\\
    a(f) \big(u_1\otimes_s...\otimes_s u_n\big)  = \frac{1}{\sqrt{n}} \sum_{k = 1}^n \langle f | u_k\rangle \big(u_1\otimes_s ...\cancel{u_k}...\otimes u_n\big).\end{cases}
\end{equation}
By linearity and density, these definitions can be extended to $D(\Ncal^{1/2})$. As the notation suggests, for $\Psi, \Phi \in D(\Ncal^{1/2})$, $\langle \Phi| a(f)\Psi\rangle = \langle a^*(f)\Phi|\Psi\rangle$. The creation and annihilation operators satisfy the Canonical Commutation Relations (CCR) on $D(\Ncal)$:
\begin{equation}
    \begin{cases}
        [a(f), a^*(g)] = \langle f|g\rangle \\
        [a(f), a(g)] = 0\\
        [a^*(f), a^*(g)] = 0.
    \end{cases}
\end{equation}
They also satisfy the following inequality (see for instance Proposition 2.1 of \cite{bruneau_bogoliubov_2007}):
\begin{equation}\label{equation_inegalite_creation_annihilation}
    \forall \Psi \in \Fcal,~\forall s\in \mathbb{R}_+,~ \forall f\in D(\omega^{-s/2}),~\|a(f)\Psi\|^2 \leq \|\omega^{-s/2}f\|^2 \langle \Psi|\dg(\omega^s) \Psi\rangle 
\end{equation}
A polynomial in the creation and annihilation operators is said to be \enquote{normal ordered} if all creation operators are to the left of annihilation operators. The normal ordering of such an expression is a normal ordered expression involving the same creation and annihilation operators, and is written as
\begin{equation}
    {:}\big(a^*(f)+a(f)\big)^2{:}~=~{:}\big(a(f)^2 + a^*(f)^2 + a^*(f)a(f) + a(f) a^*(f)\big){:}~=~a(f)^2 +a^*(f)^2 + 2 a^*(f)a(f).
\end{equation}

We call \enquote{Bogoliubov transformation} a unitary operator $\U$ on $\Fcal$ such that there exist $U:\Hcal \to \Hcal$ and $V : \Hcal^* \to \Hcal$, two bounded operators satisfying for all $g \in \Hcal$ 
    \begin{equation}\label{equation_definition_transformation_bogoliubov}
        \begin{cases}
            \U^* a^*(g) \U = a^*(Ug) + a(VJg)\\
        \U a^*(g) \U^* = a^*(U^*g) - a(J^*V^*g).
        \end{cases}
    \end{equation}
As proved e.g. in Lemma 4.4 of \cite{beyond_bogoliubov}, the states $\U\Omega$ and $\U^* \Omega$ have a finite number of particles expectation, i.e.
    \begin{equation}
        \big\langle \U \Omega | \Ncal \U \Omega \big\rangle < + \infty,~~~~~~\big\langle \U^* \Omega | \Ncal \U^* \Omega \big\rangle < + \infty.
    \end{equation}
If $\U$ is a Bogoliubov transformation associated with $U$ and $V$, for all $z\in \C$ of modulus 1, $z\U$ is also a Bogoliubov transformation associated with $U$ and $V$. This means that there is not a unique Bogoliubov transformation associated with $U$ and $V$, and that there always is a choice of phase. Let $\Gcal$ be the subset of bounded operators on $\Hcal \oplus \Hcal^*$ defined by
\begin{equation}
    \Gcal = \left\{\Vcal = \begin{pmatrix}
        U & V\\
        JVJ & JUJ^*
    \end{pmatrix},~ \Vcal^* \Scal \Vcal = \Vcal \Scal \Vcal^* = \Scal = \begin{pmatrix}
        1 & 0\\
        0 & -1
    \end{pmatrix}, ~\tr\big(VV^*\big)<+\infty  \right\}.
\end{equation}
It follows directly from the definition that $\Vcal \in \Gcal \Leftrightarrow \Vcal^* \in \Gcal$. The equation $\Vcal^* \Scal \Vcal = \Vcal \Scal \Vcal^* = \Scal$ is a symplectic condition, and the inequality $\tr\big(VV^*\big)<+\infty$ is called Shale's condition. The following Proposition gives a condition for the existence of Bogoliubov transformations (see for instance Section 1.2. of \cite{nam_diagonalization_2016}). Note that there is no canonical choice of such a Bogoliubov transformation.

\begin{proposition}[Condition for the existence of Bogoliubov transformations]\label{proposition_existence_transformation_bogolioubov}
    The linear operators $U :\Hcal \to \Hcal$ and $V:\Hcal^* \to \Hcal$ are implemented by a Bogoliubov transformation $\U$ on $\Fcal$ if and only if 
    \begin{equation}
        \Vcal := \begin{pmatrix}
            U & V\\
            JVJ & JUJ^*
        \end{pmatrix} \in \Gcal.
    \end{equation}
\end{proposition}

\subsection{Preliminary discussion}

Let us now briefly discuss some cases where the Hamiltonian formally given by (\ref{equation_normal_formal_hamiltonian}) is clearly well defined. We recall the formal Hamiltonian :
\begin{equation}\label{equation_definition_hamiltonien_formel}
    \mathbb{H}_{\lambda}^{\mathrm{f}} := \d\Gamma(\omega) + \lambda~{:}\big(a^*(f) + a(f)\big)^2{:} = \d\Gamma(\hl) + \lambda\big(a^*(f)^2 + a(f)^2\big),
\end{equation}
where $\hl = \omega + 2\lambda|f\rangle\langle f|$. The operator $\omega$ is the self-adjoint Hamiltonian of the one-particle space, that we assume to be unbounded and greater than 1, and $\lambda$ is the (positive) coupling constant.\\ 

First, we want to show that $\mathbb{H}_{\lambda}^{\mathrm{f}}$ is a self-adjoint operator when $f\in \Hcal$. We want to use the Kato-Rellich theorem, which requires us to have a relative bound on the second term of the Hamiltonian. For all $\Psi \in D(\Ncal)$,
\begin{equation}
    \begin{cases}
        \|a(f)^2\Psi\|^2 \leq  \|f\|^4 \langle \Psi|\Ncal(\Ncal-1)_+\Psi\rangle\\
        \|a^*(f)^2\Psi\|^2 \leq \|f\|^4 \langle \Psi|(\Ncal^2 + 2)^2 \Psi\rangle\\
        \|a^*(f)a(f)\Psi\|^2\leq  \|f\|^4 \langle \Psi|\Ncal^2 \Psi\rangle.
    \end{cases}
\end{equation}
Thus, for all $\Psi\in D(\Ncal)$
\begin{equation}\label{equation_inegalite_normal_creation_annihilation_nombre}
    \big\|{:}\big(a^*(f) + a(f)\big)^2{:}~\Psi\| \leq \sqrt6\|f\|^2\|\Ncal\Psi\| + \sqrt 6 \|f\|^2 \|\Psi\| \leq \sqrt6\|f\|^2\|\dg(\omega)\Psi\| + \sqrt 6 \|f\|^2 \|\Psi\|.
\end{equation}
Therefore, when $\lambda\|f\|^2 < 1/\sqrt6$,  $\mathbb{H}_{\lambda}^{\mathrm{f}}$ is a self-adjoint operator on $D(\dg(\omega))$. In general, we can use the commutator theorem (Theorem X.36 of \cite{reed_methods_1975}). Because of (\ref{equation_inegalite_normal_creation_annihilation_nombre}), the operator $\Tilde{\mathbb{H}}_{\lambda}^{\mathrm{f}} := \mathbb{H}_{\lambda}^{\mathrm{f}} + \sqrt 6 \lambda \|f\|^2 (\Ncal+1)$ is self-adjoint on $D(\dg(\omega))$, and for all $\Psi \in D\big(\dg(\omega)^{1/2}\big)$,
\begin{equation}
    \big|\langle \Psi |\big[\Tilde{\mathbb{H}}_{\lambda}^{\mathrm{f}}, \mathbb{H}_{\lambda}^{\mathrm{f}}\big] \Psi\rangle\big| = \sqrt 6 \lambda^2\|f\|^2 \big|\langle \Psi| \big[\Ncal, {:}(a^*(f) + a(f))^2{:} \big]\Psi\rangle \big| \leq c \langle \Psi| (\Ncal+1) \Psi \rangle \leq c \langle \Psi |(\dg(\omega)+1)\Psi\rangle,
\end{equation}
as $\big[\Ncal, {:}(a^*(f) + a(f))^2 \big]$ is quadratic in the creation and annihilation operators. Therefore, $\mathbb{H}_{\lambda}^{\mathrm{f}}$ is essentially self-adjoint on $D\big(\dg(\omega)\big)$. When $f\notin \Hcal$, it is not immediately clear on which domain $\mathbb{H}_{\lambda}^{\mathrm{f}}$ may be defined. ~\\

Besides, for $f\in \Hcal$, $\mathbb{H}_{\lambda}^{\mathrm{f}}$ clearly defines a quadratic form on $D\big(\dg(\omega)^{1/2}\big)$. Note however that $\mathbb{H}_{\lambda}^{\mathrm{f}}$ is not a priori bounded from below for a general $\lambda$. When $\omega^{-1/4}f \in \Hcal$, we have for all $\Psi \in D(\dg(\omega)^{1/2})$
\begin{equation}
    \big|\langle \Psi|a^*(f)^2\Psi\rangle \big| = \big|\langle \Psi| a(f)^2\Psi\rangle\big| \leq \|\Ncal^{1/2}\Psi\|\|\omega^{-1/4}f\|^2\|\Ncal^{-1/2}\dg(\omega^{1/2})\Psi\|.
\end{equation}
As $\dg(\omega^{1/2})^2 \leq \Ncal \dg(\omega)$, the quadratic forms $a(f)^2$ and $a^*(f)^2$ are well defined on the form domain of $\dg(\omega)$.  Therefore, when $\omega^{-1/4}f\in \Hcal$, $\mathbb{H}_{\lambda}^{\mathrm{f}}$ defines a quadratic form on $D(\dg(\omega)^{1/2})$. Because of the KLMN theorem (Theorem X.17 of \cite{reed_methods_1975}), when $4\lambda\|\omega^{-1/4}f\|^2 < 1$, the quadratic form $\mathbb{H}_{\lambda}^{\mathrm{f}}$ on $D(\dg(\omega)^{1/2})$ is bounded from below, and there exists a self-adjoint operator whose quadratic form is $\mathbb{H}_{\lambda}^{\mathrm{f}}$. \\

\textbf{Our goal is now to show that when $f\notin \Hcal$, there may still exist a self-adjoint operator formally defined by $\mathbb{H}_{\lambda}^{\mathrm{f}}$.}

\section{Main results}

Let $\omega$ be a self-adjoint operator on $\Hcal$, satisfying $\omega \geq 1$. As $\omega$ is self-adjoint, we can assume, using the spectral theorem, that $\Hcal = L^2(X, \mu)$ for a certain measure space $(X, \mu)$ and that $\omega$ is a real multiplication operator, i.e.
\begin{equation} 
	\forall f \in \Hcal,~\forall k \in X,~ (\omega f)(k) = \omega(k) f(k)
\end{equation}
	with $\omega(k) \in \R$. Let $\lambda \in \R_+$, we set 
\begin{equation}
	\hl := \omega + 2\lambda|f\rangle\langle f|.
\end{equation}
 For $f\in \Hcal$, $\omega^2 + 4 \lambda \omega^{1/2}|f\rangle\langle f|\omega^{1/2}$ is a positive quadratic form on $D(\omega)$, whose Friedrichs extension is a positive self-adjoint operator, and we call $\xi_{\lambda}$ its square root. When $f$ is regular enough, namely $\ln(\omega)^2f \in \Hcal$, the operator defined by (\ref{equation_definition_hamiltonien_formel}) can be diagonalized, and its diagonalization is of the form $\d \Gamma(\xi_{\lambda}) + C$. Note that this does not require any renormalization.

\begin{proposition}\label{theorem_0}
    Let $f\in \Hcal$ be a real-valued function such that $\ln(\omega)^2f\in \Hcal$. Let $\Hl$ be the self-adjoint extension of the operator defined by $\dg(\omega) + \lambda {:}\big(a^*(f) + a(f)\big)^2{:}$ on $D\big(\dg(\omega)\big)$. Then, there exists a Bogoliubov transformation $\Ul$ and a self-adjoint operator $\xil$ on $\Hcal$ such that $\xi_{\lambda} - h_{\lambda}$ is trace class and
    \begin{equation} \label{equation_hamiltonien_proposition}
        \U_{\lambda}^* \Hl \Ul = \mathrm{d}\Gamma\big(\xi_{\lambda}\big)+ \frac{1}{2}\tr\big(\xi_{\lambda} - h_{\lambda}\big),
    \end{equation}
    where $\xil$ depends explicitly on $\omega$, $f$ and $\lambda$:
    \begin{equation}\label{equation_definition_xil}
        \xil = \Big(\omega^2 + 4\lambda\omega^{1/2}|f\rangle\langle f|\omega^{1/2}\Big)^{1/2}.
    \end{equation}
    In particular, the domain of self-adjointness of $\Hl$ is $\Ul D\big(\dg(\xil)\big)$.
\end{proposition}

The assumption  that $\ln(\omega)^2f \in \Hcal$ will be discussed in detail when it naturally appears in the proof.

\begin{remark}[When $f$ is not real valued]\label{remarque_cas_f_non_reel}
    When $f$ is not real valued, there exists a unitary transformation that maps the problem to a case where the function is real valued and positive. Indeed, if we set $\varphi = f/|f|$ and still denote by $\varphi$ the multiplication by $\varphi$, we have
    \begin{equation}
        \Gamma(\varphi)^*\mathbb{H}_{\lambda}(f)\Gamma(\varphi) = \mathbb{H}_{\lambda}(|f|).
    \end{equation}
    This will also be true for Theorem \ref{theorem1} and \ref{theorem2}.
\end{remark}

\begin{remark}
    It is clear from (\ref{equation_hamiltonien_proposition}) that the quadratic form associated to $\Hl$ is bounded from below for all $\lambda\geq 0$.
\end{remark}

\noindent The proof of Proposition \ref{theorem_0} will be given in Section \ref{section_preuve_proposition}.~\\

The aim of this paper is to use the diagonalized version of the Hamiltonian to extend the definition of $\Hl$ with $f\notin \Hcal$. However, when $f\notin \Hcal$, it is not immediately clear how to define $\xi_{\lambda}$, so we apply the renormalization presented in \cite{derezinski_renormalization_2002} (see also \cite{kiselev_rank_1995}). If $\omega^{1/2}f\in \Hcal$, we denote by $R_z(\lambda, f)$ the resolvent of the operator $\xil^2 = \omega^2 + 4 \lambda|\omega^{1/2}f\rangle \langle \omega^{1/2}f|$ with $z\in \C \setminus \R$. We have
\begin{equation}\label{equation_resolvante}
    R_z(\lambda, f) = \big(\omega^2 - z\big)^{-1} - \frac{4\lambda }{1 + 4\lambda \langle f | \omega(\omega^2 -z)^{-1}f\rangle}\omega^{1/2}\big(\omega^2 - z)^{-1}|f\rangle\langle f|\omega^{1/2}\big(\omega^2 - z)^{-1},
\end{equation}
and we can notice that this expression is still valid when $\omega^{-1/2}f\in \Hcal$. Moreover, the operator defined by (\ref{equation_resolvante}) is the resolvent of a positive self-adjoint operator, which means that we can still define $\xi_{\lambda}^2$ from its resolvent, and $\xi_{\lambda}$ as its positive square root. More detail will be given in Section \ref{ssection_existence_regulier}. \\

\noindent It remains to check wether the constant in (\ref{equation_hamiltonien_proposition}) is still finite when $f\notin \Hcal$. When $\omega^{1/2}f\in \Hcal$, we have 
\begin{equation}
    \tr\big(\xi_{\lambda} - h_{\lambda}\big) = \frac{-32\lambda^2}{\pi} \int_0^{+\infty} \frac{t^2\|R_{-t^2}(\omega^2)\omega^{1/2}f\|^2 \big\langle f\big| R_{-t^2}(\omega^2) \omega f\big\rangle }{1 + 4\lambda\big\langle f | R_{-t^2}(\omega^2) \omega f \big\rangle}\d t,
\end{equation}
and we can see that the right-hand side is still finite when $\omega^{-1/4}f\in \Hcal$ (see Section \ref{subsection_trace}). Indeed, 
\begin{equation}
	\|R_{-t^2}(\omega^2)\omega^{1/2}f\|^2 \leq t^{-5/4}\|\omega^{-1/4}f\|^2,
\end{equation}
and
\begin{equation}
	\int_0^{+\infty} t^{-1/2}\big\langle f\big| R_{-t^2}(\omega^2) \omega f\big\rangle \d t \leq \frac{\pi}{\sqrt{2}} \|\omega^{-1/4}f\|^2.
\end{equation} On the contrary, when $\omega^{-1/4}f \notin \Hcal$, the right-hand side is not finite a priori, and we need to renormalize the energy. Now, we can state our first theorem.

\begin{theorem}[No renormalization and renormalization of the energy]\label{theorem1}~\\
    Let $f\in D(\omega^{-1/2})$ such that $\omega^{-1/2}f$ is real valued, and let $f_n = \mathds{1}_{\omega \leq n} f \in D(\omega^{3/2})$. Assume $\lambda \geq 0$ and let $\xil$ be the square-root of the self-adjoint operator whose resolvent is defined by \emph{(}\ref{equation_resolvante}\emph{)}. Then $\omega^{-1/2}f_n \to \omega^{-1/2}f$ in $\Hcal$.
        \begin{enumerate}
            \item If $\omega^{-1/4}f \in \Hcal$, there exists a Bogoliubov transformation $\U_{\lambda}$ such that 
            \begin{equation}
                 \Ham_{\lambda}  := \U_{\lambda}\Big(\dg\big(\xi_{\lambda}\big) + \frac{1}{2}\tr_{\infty}\big( \xi_{\lambda} - h_{\lambda}\big)\Big)\U_{\lambda}^* = \lim_{n \to \infty} \Big( \dg(\omega) + \lambda {:}\big(a^*(f_n) + a(f_n)\big)^2{:}\Big) ,
            \end{equation}
            where 
            \begin{equation} \label{equation_trace_formelle}
                \tr_{\infty}\big(\xi_{\lambda} - h_{\lambda}\big) := \frac{-32\lambda^2}{\pi} \int_0^{+\infty} \frac{t^2\|R_{-t^2}(\omega^2)\omega^{1/2}f\|^2 \big\langle f\big| R_{-t^2}(\omega^2)\omega f\big\rangle }{1 + 4\lambda\big\langle f | R_{-t^2}(\omega^2) \omega f \big\rangle}\d t > -\infty.
            \end{equation}
    
            \item Assume $\omega^{-1/4}f \notin \Hcal$, we define $\xi_{\lambda, n}$ by \emph{(}\ref{equation_definition_xil}\emph{)} for $f_n$, and let 
            \begin{equation}
            	E_n = \frac{1}{2}\tr\big(\xi_{\lambda, n} - h_{\lambda, n}\big) \to -\infty.
            \end{equation}
         	Then, there exists a Bogoliubov transformation $\U_{\lambda}$ such that 
            \begin{equation}
                \Ham_{\lambda} := \U_{\lambda} \dg\big(\xi_{\lambda}\big) \U_{\lambda}^* = \lim_{n \to \infty} \Big( \dg(\omega) + \lambda {:}\big(a^*(f_n) + a(f_n)\big)^2{:} - E_n\Big).
            \end{equation}
        \end{enumerate}
        In both cases, $\Hl$ is self-adjoint on $\mathbb{U}_{\lambda}D(\dg(\xi_{\lambda}))$, and the limits are in the strong resolvent sense.
\end{theorem}
\begin{remark}[Convergence of the spectrum]
	In Theorem \ref{theorem1}, the convergence of $\dg(\xi_{\lambda, n})$ to $\dg(\xi_{\lambda})$ is in the norm resolvent sense. As $\mathbb{U}_{\lambda, n}$ and $\mathbb{U}_{\lambda}$ are unitary operators, this implies the convergence of the spectrum (see Theorem VIII.24 of \cite{reed_i_1980}).
\end{remark}
\begin{remark}[More general $f_n$]
	Theorem \ref{theorem1} remains true for a general sequence $(f_n) \in D(\omega^{1/2})^{\mathbb{N}}$ satisfying
	\begin{equation} \label{equation_remarque_fn}
		\begin{cases}
			\omega^{-1/2}f_n \to \omega^{-1/2}f\\
			|f_n| \leq |f|\\
			f_n(k) \to f(k)~\text{for almost every } k.			
		\end{cases}
	\end{equation}
\end{remark}

Now, we want to define $\xi_{\lambda}$ when $\omega^{-1/2}f\notin \Hcal$. Note that if we simply use (\ref{equation_resolvante}), there will be a divergence in the denominator. Therefore, we have to renormalize the coupling constant. Once again, as in \cite{derezinski_renormalization_2002}, for a given $\lambda \in \mathbb{R}_-$, we set $\tilde{\lambda}^{-1} = \lambda^{-1} - 4\langle f |\omega^{-1}f\rangle$. We have
\begin{equation}\label{equation_resolvante_renormalisee}
	\tilde{R}_z(\lambda, f) := R_z(\tilde{\lambda}, f) = \big(\omega^2 - z\big)^{-1} -\frac{4\lambda}{1 + 4\lambda z \langle f |\omega^{-1}(\omega^2 - z)^{-1} f\rangle}\omega^{1/2}(\omega^2 - z)^{-1}|f\rangle\langle f| \omega^{1/2}(\omega^2 - z)^{-1},
\end{equation}
and we can notice that this expression is still valid when $\omega^{-3/2}f \in \Hcal$. Moreover, as long as $\lambda \leq 0$ the operator defined by (\ref{equation_resolvante_renormalisee}) is in particular well defined for $z\in \mathbb{R}_-$ and is thus the resolvent of a positive operator, which means that we can still define $\xil$ from the resolvent of its square. More details will be given in Section \ref{ssection_existence_irregulier}. We can then state our second theorem.
    
\begin{theorem}[Renormalization of the energy and the charge] \label{theorem2}~\\
	Let $f\in D(\omega^{-3/2})$ such that $\omega^{-3/2}f$ is real valued, and let $f_n = \mathds{1}_{\omega \leq n} f \in D(\omega^{3/2})$. Assume $\lambda \leq 0$, let $\xil$ be the square-root of the self-adjoint operator whose resolvent is defined by \emph{(}\ref{equation_resolvante_renormalisee}\emph{)} and let $\lambda_n \in\mathbb{R}_-$ be the negative number whose inverse is
	\begin{equation}
		\lambda_n^{-1} = \lambda^{-1} - 4\langle f_n|\omega^{-1}f_n\rangle \to -\infty.
	\end{equation} 
	Moreover, let $\xi_{\lambda_n, n}$ be the operator defined by \emph{(}\ref{equation_definition_xil}\emph{)} for $\lambda_n$ and $f_n$ and
	\begin{equation}
		E_n = \frac{1}{2}\tr\big(\xi_{\lambda_n, n} - h_{\lambda_n}\big).
	\end{equation}
	Then, there exists a sequence of Bogolioubov transformations $\U_{\lambda_n, n}$ such that
    \begin{equation}
        \dg(\xil) = \lim_{n\to \infty} \U_{\lambda_n,n}^* \Big(\dg(\omega) + \lambda_n{:}\big(a^*(f_n) + a(f_n)\big)^2{:} - E_n\Big)\U_{\lambda_n, n},
    \end{equation}
    where the limit is in the norm resolvent sense.
\end{theorem}

\begin{remark}[Non existence of Bogoliubov transformations]
    Under the assumptions of Theorem \ref{theorem2}, $\Ul$ cannot be defined in general, which means that we cannot come back to $\Hl$. To be more specific, in Proposition \ref{proposition_nonexistence_transformations}, we show two different results:
    \begin{itemize}
        \item we prove that for all $f\in D(\omega^{-3/2})\setminus D(\omega^{-1})$, $\tr(V_{\lambda}V_{\lambda}^*) = + \infty$ and therefore there does not exist a Bogoliubov transformation associated with $U_{\lambda}$ and $V_{\lambda}$;
        \item we give an example of $f\in D(\omega^{-1})\setminus D(\omega^{-1/2})$ such that $\tr(V_{\lambda}V_{\lambda}^*) = +\infty$. 
    \end{itemize}
\end{remark}

\begin{remark}[More general $f_n$]
	Theorem \ref{theorem2} remains true for a general sequence $(f_n) \in D(\omega^{1/2})^{\mathbb{N}}$ satisfying
	\begin{equation}
			\omega^{-3/2}f_n \to \omega^{-3/2}f.
	\end{equation}
\end{remark}

Let us now summarise Theorem \ref{theorem1} and \ref{theorem2}:
\begin{table}[h]
	\centering
	\begin{tabular}{|c|c|c|}\hline
		Regularity of $f$ & Renormalization & Bogoliubov transformations\\ \hline
		$f \in D(\omega^{-1/4})$ & none & existence\\
		$f\in D(\omega^{-1/2})\setminus D(\omega^{-1/4})$ & energy & existence\\
		$f\in D(\omega^{-3/2})\setminus D(\omega^{-1/2})$ &  energy and charge & non-existence\\\hline
	\end{tabular}
	\caption{Summary of the theorems}
	\label{tab:placeholder}
\end{table}

\begin{remark}[Concrete examples]
	Now, let us discuss specific models motivated by the Pauli-Fierz Hamiltonian. If we deduce $\mathbb{H}_{\lambda}$ from the Pauli-Fierz Hamiltonian with $p = 0$ and neglecting polarization, we find $\omega(k) = |k|$ and $f(k) = |k|^{-1/2}$. To apply our Theorem, we add a cut-off on low frequencies and take $\omega(k) = \max(1, |k|)$ and $f(k) = |k|^{-1/2}\mathds{1}_{|k|\geq 1}$. Then, when $d = 2$ or $d = 3$, we have $\omega^{-3/2}f\in \Hcal$ but $\omega^{-1/2}f\notin \Hcal$, which means that the coupling constant has to be renormalized. However, it is well-known in other models (see \cite{griesemer_domain_2018} for instance) that $p^2$ has a regularizing effect. To take this into consideration, let $\omega(k) = \max(1, |k| + |k|^2)$. Now, with this new choice of $\omega$, $\omega^{-1/2}f\in \Hcal$ when $ d = 2 $, and the renormalization of the charge only appears for $d=3$. More accurately, when $d = 3$, we have $|\omega^{-1/2}(k)f(k)|^2 \sim |k|^{-3}$, which means that the renormalization of the coupling constant is logarithmic. This is coherent with the equation (19.85) of \cite{spohn_dynamics_2004}, wich implies that $\alpha \sim 1/\log(\Lambda)$, $\alpha$ being the fine structure constant and $\Lambda$ the cut-off.
\end{remark}

The proof of Proposition \ref{theorem_0} is given in Section \ref{section_preuve_proposition}. It relies heavily on the proof of Theorem 2 from \cite{nam_diagonalization_2016}. In Section \ref{section_existence_xil}, we give an overview of the proof of the existence of $\xi_{\lambda}$. The proof of Theorem \ref{theorem1} is given in Section \ref{section_preuve_theoreme1}. First, we show the existence of the integral defined by (\ref{equation_trace_formelle}) when $\omega^{-1/4}f\in \Hcal$. Then we prove the existence of Bogoliubov transformations, and finally, we show the convergence of $\dg(\xi_{\lambda, n})$ and $\mathbb{U}_{\lambda, n}$. In Section \ref{section_preuve_theoreme2}, we prove Theorem \ref{theorem2}, namely the convergence of $\dg(\xi_{\lambda, n})$, and we discuss the existence of Bogoliubov transformations.

\section{Proof of Proposition \ref{theorem_0}}\label{section_preuve_proposition}

We separate the proof of Proposition \ref{theorem_0} in several lemmas. First, in Lemma \ref{lemme_nam_etc}, we show (using \cite{nam_diagonalization_2016}) the renormalization of the Hamiltonian under certain assumptions that we then check in Lemma \ref{lemme_verification_hypothese_trace_class}. Finally, we show in Lemma \ref{lemme_calcul_constante} that the constant we find in Lemma \ref{lemme_nam_etc} is indeed $\tr(\xi_{\lambda} - h_{\lambda})/2$.

\begin{definition}\label{definition_preuve_proposition}
	We denote the complex conjugation of a function of $\Hcal = L^2(X, \mu)$ by  $\mathfrak{c} : \Hcal \to \Hcal$. We remark that $\mathfrak{c}$ is anti-linear and satisfies $\mathfrak{c}^* = \mathfrak{c}$. Then, let
	\begin{equation}\label{equation_definition_preuve_proposition}
		k'_{\lambda} := 2\lambda |f\rangle \langle f | ~~~~~~~~~~\text{and} ~~~~~~~~~~ h_{\lambda} := \omega + k'_{\lambda}.
	\end{equation}
Moreover, recall that $J : |\psi\rangle \mapsto \langle \psi|$ is the canonical isometry between $\Hcal$ and $\Hcal^*$ and let $k_{\lambda}: \Hcal^* \to \Hcal$ be the operator satisfying 
\begin{equation}
	k_{\lambda} J = k_{\lambda}' \mathfrak{c}.
\end{equation}
\end{definition}
Before we can state and prove Lemma \ref{lemme_nam_etc}, we introduce first the following lemma.
\begin{lemma}[Diagonalization of the block operator associated with the Hamiltonian]\label{lemme_diagonalisation}
	Let $\omega$ be a self-adjoint operator on $\Hcal$ commuting with $\mathfrak{c}$. Let $k : \Hcal^* \to \Hcal$ and let $k'$ be a self-adjoint operator on $\Hcal$, commuting with $\mathfrak{c}$ and satisfying
	\begin{equation}
		kJ = k'\mathfrak{c}.
	\end{equation}
	We set
	\begin{equation}\label{equation_definition_lemme_diagonalisation}
		\left\{  \begin{array}{llllll}
			\xi = \Big(\omega^{1/2}(\omega + 2k')\omega^{1/2}\Big)^{1/2}\\
			U = (\omega^{1/2} \xi^{-1/2} + \omega^{-1/2} \xi^{1/2})/2\\
			V = \mathfrak{c}(\omega^{1/2} \xi^{-1/2} - \omega^{-1/2} \xi^{1/2})J^*/2\\
			h = \omega + k',
		\end{array}
		\right.
	\end{equation}
	and define
	\begin{equation}\label{equation_definition_lemme_diagonalisation2}
		\Acal = \begin{pmatrix}
			h & k\\
			k^*& JhJ^*
		\end{pmatrix},~~~~~~ \Vcal^* = \begin{pmatrix}
			U & V\\
			JVJ & JUJ^*
		\end{pmatrix}.
	\end{equation}
	Then,
	\begin{equation}
		\Vcal \Acal \Vcal^* = \begin{pmatrix}
			\xi & 0\\
			0 & J\xi J^*
		\end{pmatrix}.
	\end{equation}
\end{lemma}
This lemma is a particular case of Theorem 1 of \cite{nam_diagonalization_2016}, with an explicit expression of $\xi$ and $\Vcal$. The proof is straightforward and the computations are similar to those in the fourth part of \cite{grech_excitation_2013}.
\begin{remark}\label{remarque_lemme_diagonalisation}
	In the following, we will use the notations introduced in \eqref{equation_definition_lemme_diagonalisation} and \eqref{equation_definition_lemme_diagonalisation2} with subscript $\lambda$ for the particular choice $k' = k'_{\lambda}$.
\end{remark}

\begin{lemma}[Diagonalization of the Hamiltonian]\label{lemme_nam_etc}
	Let $f\in \Hcal$ be a real-valued function and $\Hl$ the self-adjoint extension of the operator defined by $\dg(\omega) + \lambda {:}\big(a^*(f) + a(f)\big)^2{:}$ on $D\big(\dg(\omega)\big)$. Assume that $h_{\lambda}^{1/2} V_{\lambda}V_{\lambda}^*h_{\lambda}^{1/2}$ and $k_{\lambda}^*V_{\lambda} J U_{\lambda}^* J^*$ (see Remark \ref{remarque_lemme_diagonalisation}) are trace class. Then, we have
	\begin{equation}\label{equation_cas_regulier_avant_simplification}
		\U_{\lambda}^* \Ham_{\lambda} \U_{\lambda} = \dg(\xi_{\lambda}) + \tr(h_{\lambda}^{1/2}V_{\lambda}V_{\lambda}^*h_{\lambda}^{1/2}) + \mathfrak{R}\tr(k_{\lambda}^*V_{\lambda} J U_{\lambda}^* J^*),
	\end{equation}	
	with 
	\begin{equation}\label{equation_preuve_proposition_xilambda}
		\xi_{\lambda} = \Big(\omega^2 + 4\lambda \omega^{1/2}|f\rangle\langle f|\omega^{1/2}\Big)^{1/2}. 
	\end{equation}
\end{lemma}
 
\begin{proof}
For an orthonormal basis $(u_j)$ of $\Hcal$, we have, using the notations introduced in Definition~\ref{definition_preuve_proposition},
\begin{equation}
    \Ham_{\lambda} = \dg(h_{\lambda}) + \frac{1}{2}\sum_{j, l \in \N} \Big( \langle u_j | k_{\lambda}J u_l\rangle a^*(u_j) a^*(u_l) + \overline{\langle u_j | k_{\lambda}Ju_l\rangle} a(u_j) a(u_l)\Big).
\end{equation}
This is the general expression of a quadratic Hamiltonian. From the proof of Theorem 2 from \cite{nam_diagonalization_2016}, we know that if $h_{\lambda}^{1/2} V_{\lambda}V_{\lambda}^*h_{\lambda}^{1/2}$ and $k_{\lambda}^*V_{\lambda} J U_{\lambda}^* J^*$ are trace class, then we have \eqref{equation_cas_regulier_avant_simplification}, the self-adjoint operator $\xi_{\lambda}$ being given by the diagonalization of the operator \begin{equation}
	\mathcal{A}_{\lambda} := \begin{pmatrix}
		h_{\lambda} & k_{\lambda}\\
		k_{\lambda}^*& Jh_{\lambda}J^*
	\end{pmatrix}.
\end{equation}
To get \eqref{equation_preuve_proposition_xilambda}, we want to use Lemma \ref{lemme_diagonalisation} with $k' = k'_{\lambda}$. As $f$ is real valued, $k_{\lambda}'$ and $\mathfrak{c}$ indeed commute. Moreover, as $\omega$ is a real multiplication operator, it is clear that $\omega$ and $\mathfrak{c}$ commute. Therefore, the assumptions of Lemma \ref{lemme_diagonalisation} are indeed satisfied, concluding the proof of Lemma \ref{lemme_nam_etc}.
\end{proof}

Now, to prove Proposition \ref{theorem_0}, we have to check that $h_{\lambda}^{1/2} V_{\lambda}V_{\lambda}^*h_{\lambda}^{1/2}$ and $k_{\lambda}^*V_{\lambda} J U_{\lambda}^* J^*$ are trace class, and explicitly compute the constant $\tr(h_{\lambda}^{1/2}V_{\lambda}V_{\lambda}^*h_{\lambda}^{1/2}) + \mathfrak{R}\tr(k_{\lambda}^*V_{\lambda} J U_{\lambda}^* J^*)$.

\begin{lemma}[Verifying the assumptions of Lemma \ref{lemme_nam_etc}]\label{lemme_verification_hypothese_trace_class}
	Let $f\in \Hcal$. Let $k_{\lambda}'$ $h_{\lambda}$, $U_{\lambda}$ and $V_{\lambda}$ be defined in \eqref{equation_definition_preuve_proposition} and \eqref{equation_definition_lemme_diagonalisation} respectively. Then, $h_{\lambda}^{1/2} V_{\lambda}V_{\lambda}^*h_{\lambda}^{1/2}$ and $k_{\lambda}^*V_{\lambda} J U_{\lambda}^* J^*$ are trace class.
\end{lemma}

As $\xi_{\lambda}$ is defined as the square root of a rank one perturbation, we use the following lemma to have an explicit formula for $\xi_{\lambda}$ and $\xi_{\lambda}^{-1}$.

\begin{lemma}[Integral formulas for roots of a rank one perturbation]\label{lemme_formules_integrales_racine_perturbation}
    Let $A$ be a self-adjoint positive operator, $\psi \in \Hcal$ and $\alpha \in \R$ such that $A+\alpha|\psi\rangle\langle \psi|$ is positive (this is satisfied if $\alpha \in \R_+$). We denote $T_{\alpha} = A + \alpha |\psi\rangle\langle \psi|$ and we recall that $R_z(A)$ denotes the resolvent of $A$, then
    \begin{equation}\label{equation_racine_carree_perturbation}
        T_{\alpha}^{1/2} = A^{1/2} + \frac{2\alpha}{\pi}\int_0^{+\infty}\frac{t^2 |R_{-t^2}(A)\psi\rangle\langle R_{-t^2}(A) \psi|}{1 + \alpha \langle \psi | R_{-t^2}(A) \psi\rangle}\d t,
    \end{equation}

    \begin{equation}\label{equation_inverse_racine_carree_perturbation}
        T_{\alpha}^{-1/2} = A^{-1/2} - \frac{2\alpha}{\pi}\int_0^{+\infty} \frac{|R_{-t^2}(A)\psi\rangle\langle R_{-t^2}(A) \psi|}{1 + \alpha \langle \psi | R_{-t^2}(A) \psi\rangle}\d t,
    \end{equation}

    \begin{equation}\label{equation_racine_quatrieme_perturbation}
        T_{\alpha}^{1/4} = A^{1/4} + \frac{2\sqrt{2}\alpha}{\pi}\int_0^{+\infty}\frac{t^4 |R_{-t^4}(A)\psi\rangle \langle R_{-t^4}(A)\psi|}{1 + \alpha \langle \psi | R_{-t^4}(A)\psi\rangle}\d t,
    \end{equation}
    and
    \begin{equation}\label{equation_inverse_racine_quatrieme_perturbation}
        T_{\alpha}^{-1/4} = A^{-1/4} - \frac{2\sqrt{2}\alpha}{3\pi}\int_0^{+\infty}\frac{ |R_{-t^{4/3}}(A)\psi\rangle \langle R_{-t^{4/3}}(A)\psi|}{1 + \alpha \langle \psi | R_{-t^{4/3}}(A)\psi\rangle}\d t.
    \end{equation}
    In particular, if $A\geq 1$, this implies that $T_{\alpha}^{1/2} - A^{1/2}$ is trace class and that 
    \begin{equation}\label{equation_trace_racine_carree_perturbation}
        \tr\big(T_{\alpha}^{1/2} - A^{1/2}\big) = \frac{2\alpha}{\pi} \int_0^{+\infty} \frac{t^2\big\| R_{-t^2}(A)\psi\big\|^2}{1 + \alpha\big\langle \psi | R_{-t^2}(A)\psi\big\rangle} \d t.
    \end{equation}
\end{lemma}

This lemma is a generalization of Proposition 7.5 of \cite{ravn_christiansen_random_2021}, and a proof is given in the Appendix. 

\begin{proof}[Proof of Lemma \ref{lemme_verification_hypothese_trace_class}] Let us explicitly bound the trace of $h_{\lambda}^{1/2} V_{\lambda}V_{\lambda}^*h_{\lambda}^{1/2}$. We have
	\begin{equation}\label{equation_VVetoile}
		4 V_{\lambda} V_{\lambda}^* = \omega^{1/2}\xi_{\lambda}^{-1}\omega^{1/2} - 1 + \omega^{-1/2}\xi_{\lambda} \omega^{-1/2} - 1.
	\end{equation} 
Using (\ref{equation_racine_carree_perturbation}), we have 
\begin{equation}\label{equation_xilambda_integrale}
    \omega^{-1/2}\xi_{\lambda} \omega^{-1/2} - 1 = \frac{8\lambda}{\pi}\int_0^{+\infty}\frac{t^2 |R_{-t^2}(\omega^2)f\rangle \langle R_{-t^2}(\omega^2)f|}{1 + 4\lambda \langle f | \omega R_{-t^2 }(\omega^2) f\rangle} \d t,
\end{equation}
and hence
\begin{equation}
    h_{\lambda}^{1/2} \big(\omega^{-1/2}\xi_{\lambda} \omega^{-1/2} - 1\big) h_{\lambda}^{1/2} = \frac{8\lambda}{\pi}\int_0^{+\infty}\frac{t^2 |h_{\lambda}^{1/2} R_{-t^2}(\omega^2) f\rangle \langle h_{\lambda}^{1/2} R_{-t^2}(\omega^2) f|}{1 + 4\lambda \langle f | \omega R_{-t^2 }(\omega^2) f\rangle} \d t.
\end{equation}
This shows that $h_{\lambda}^{1/2} \big(\omega^{-1/2}\xi_{\lambda} \omega^{-1/2} - 1\big) h_{\lambda}^{1/2}$ is a self-adjoint positive operator. For a given orthonormal basis $(u_j)$ of $\Hcal$, as 
\begin{equation}
\int_0^{+\infty}\frac{a t^2}{(a^2 + t^2)^2}\d t = \frac{\pi}{4}
\end{equation}
for $a\in \R_+$, we have 
\begin{align}
    \sum_{j = 1}^{+\infty} \big\langle u_j | h_{\lambda}^{1/2} \big(\omega^{-1/2}\xi_{\lambda} \omega^{-1/2} - 1\big) h_{\lambda}^{1/2} | u_j\big\rangle &= \frac{8\lambda}{\pi} \int_0^{+\infty} \frac{t^2 \big\| h_{\lambda}^{1/2}R_{-t^2}(\omega^2)f\big\|^2}{1 + 4\lambda\langle f | \omega R_{-t^2 }(\omega^2) f\rangle} \d t \notag\\
    &\leq \frac{8\lambda}{\pi}\Big\langle f \Big| \int_0^{+\infty} \frac{t^2 \omega}{\big(\omega^2 + t^2)^2}\d t \Big| f\Big\rangle + \frac{16\lambda^2}{\pi}\int_0^{+\infty}t^2 \big|\langle f | R_{-t^2}(\omega^2) f\rangle\big|^2 \d t \notag \\
    &\leq 2\lambda \|f\|^2 + 8\lambda^2 \|f\|^4 < + \infty.
\end{align}
Likewise, using (\ref{equation_inverse_racine_carree_perturbation}),
\begin{equation}
    h_{\lambda}^{1/2}\big(\omega^{1/2} \xi_{\lambda}^{-1} \omega^{1/2} - 1\big) h_{\lambda}^{1/2} = - \frac{8\lambda}{\pi}\int_0^{+\infty} \frac{ |h_{\lambda}^{1/2} R_{-t^2}(\omega^2) \omega f\rangle \langle h_{\lambda}^{1/2} R_{-t^2}(\omega^2)\omega  f|}{1 + 4\lambda \langle f | \omega R_{-t^2 }(\omega^2) f\rangle} \d t,
\end{equation}
and as 
\begin{equation}
\int_0^{+\infty} \frac{a^3}{(a^2 +t^2)^2}\d t = \frac{\pi}{4}
\end{equation}
for $a\in \R_+$, we find
\begin{align}
    \sum_{j = 1}^{+\infty} \big\langle u_j | h_{\lambda}^{1/2} \big(1 - \omega^{1/2}\xi_{\lambda}^{-1} \omega^{1/2}\big) h_{\lambda}^{1/2} | u_j\big\rangle &= \frac{8\lambda}{\pi} \int_0^{+\infty} \frac{\big\| h_{\lambda}^{1/2}R_{-t^2}(\omega^2)\omega f\big\|^2}{1 + 4\lambda\langle f | \omega R_{-t^2 }(\omega^2) f\rangle} \d t \notag\\
    &\leq \frac{8\lambda}{\pi}\Big\langle f \Big| \int_0^{+\infty} \frac{\omega^3}{\big(\omega^2 + t^2)^2}\d t \Big| f\Big\rangle + \frac{16\lambda^2}{\pi}\int_0^{+\infty}\big|\langle f | R_{-t^2}(\omega^2)\omega f\rangle\big|^2 \d t \notag \\
    &\leq 2\lambda \|f\|^2 + 8\lambda^2 \|f\|^4 < + \infty,
\end{align}
which shows that $h_{\lambda}^{1/2} V_{\lambda} V_{\lambda}^* h_{\lambda}^{1/2}$ is indeed trace class. The proof for $k_{\lambda}^* V_{\lambda} J U_{\lambda}^*J^*$ is quite similar. We have 
\begin{equation}
    V_{\lambda}JU_{\lambda}^* J^* = \mathfrak{c}\big(\omega^{1/2}\xi_{\lambda}^{-1}\omega^{1/2} - \omega^{-1/2}\xi_{\lambda}\omega^{-1/2}\big)J^*, 
\end{equation}
and we also find 
\begin{equation} 
	\big|\tr(k_{\lambda}^* V_{\lambda} J U_{\lambda}^*J^*)\big| \leq C\big(\lambda \|f\|^2 + \lambda^2 \|f\|^4\big),
\end{equation} which concludes the proof of Lemma \ref{lemme_verification_hypothese_trace_class}.
\end{proof}

Now, we want to explicitly compute the constant in (\ref{equation_cas_regulier_avant_simplification}).
\begin{lemma}[Computation of the constant term of the Hamiltonian]\label{lemme_calcul_constante}
	Let $f\in \Hcal$ such that $\ln(\omega)^2 f\in \Hcal$, we have
	\begin{equation}
		C_{\lambda} := \tr(h_{\lambda}^{1/2}V_{\lambda}V_{\lambda}^*h_{\lambda}^{1/2}) + \mathfrak{R}\tr(k_{\lambda}^*V_{\lambda} J U_{\lambda}^* J^*) = \frac{1}{2}\tr\big(\xi_{\lambda} - h_{\lambda}\big).
	\end{equation}
\end{lemma}
\begin{proof} We have
\begin{equation}
    4 C_{\lambda}= \frac{8\lambda}{\pi}\int_0^{+\infty} \frac{\langle R_{-t^2}(\omega^2)f|(t^2 h_{\lambda} - \omega h_{\lambda} \omega - t^2 k_{\lambda}' - \omega k_{\lambda}'\omega )R_{-t^2}(\omega^2)f\rangle}{1 + 4\lambda \langle f |\omega R_{-t^2}(\omega^2)f\rangle} \d t. 
\end{equation}
On the one hand, 
\begin{align}
    \frac{8\lambda}{\pi}\int_0^{+\infty} \frac{\langle R_{-t^2}(\omega^2)f|(t^2 h_{\lambda} - t^2 k_{\lambda}')|R_{-t^2}(\omega^2)f\rangle}{1 + 4\lambda \langle f |\omega R_{-t^2}(\omega^2)f\rangle} \d t 
    &= \frac{8\lambda}{\pi}\int_0^{+\infty} \frac{t^2\langle R_{-t^2}(\omega^2)f|\omega R_{-t^2}(\omega^2)f\rangle}{1 + 4\lambda \langle f |\omega R_{-t^2}(\omega^2)f\rangle} \d t \notag\\
    &= \tr(\xi_{\lambda} - \omega).
\end{align}
On the other hand,
\begin{align}
    \frac{8\lambda}{\pi}\int_0^{+\infty} \frac{\langle R_{-t^2}(\omega^2)f|(\omega h_{\lambda} \omega + \omega k_{\lambda}'\omega )R_{-t^2}(\omega^2)f\rangle}{1 + 4\lambda \langle f |\omega R_{-t^2}(\omega^2)f\rangle} \d t  &= \frac{8\lambda}{\pi}\int_0^{+\infty} \frac{\langle R_{-t^2}(\omega^2)f|\omega^{1/2} \xi_{\lambda}^2 \omega^{1/2} R_{-t^2}(\omega^2)f\rangle}{1 + 4\lambda \langle f |\omega R_{-t^2}(\omega^2)f\rangle} \d t \notag\\
    &= \tr\Big((\xi_{\lambda} \omega^{-1/2})(\omega^{1/2}\xi_{\lambda}^{-1}\omega^{1/2} - 1)(\xi_{\lambda}\omega^{-1/2})^*\Big) \notag\\
    &= \tr(\xi_{\lambda} - \xi_{\lambda} \omega^{-1}\xi_{\lambda}).
\end{align}
Therefore, we have 
\begin{equation}
    4 C_{\lambda} = \tr\big(2\xi_{\lambda} - \omega - \xi_{\lambda}\omega^{-1}\xi_{\lambda}\big).
\end{equation}
As 
\begin{equation} 
	\omega^{-1/2}\xi_{\lambda}^2 \omega^{-1/2} = \omega + 4\lambda|f\rangle\langle f|,
\end{equation} 
we still have to show that $\tr(\xi_{\lambda}\omega^{-1}\xi_{\lambda} - \omega^{-1/2}\xi_{\lambda}^2\omega^{-1/2}) = 0$ to conclude. Using (\ref{equation_inverse_racine_carree_perturbation}), we have 
\begin{multline}\label{equation_xi_omega_xi}
    \xi_{\lambda}\omega^{-1}\xi_{\lambda} =  \left(\frac{8\lambda}{\pi}\right)^2 \int_0^{+\infty}\int_0^{+\infty} \frac{t^2 s^2 \langle R_{-t^2}(\omega^2) f| R_{-s^2}(\omega^2) f\rangle |\omega^{1/2}R_{-t^2}(\omega^2) f\rangle \langle \omega^{1/2}R_{-s^2}(\omega^2)f|}{\big(1 + 4\lambda \langle f | \omega R_{-t^2}(\omega^2)f\rangle \big)\big(1 + 4 \lambda \langle f |\omega R_{-s^2}(\omega^2) f\rangle\big)}\d s \d t \\ + 2\frac{8\lambda}{\pi}\int_0^{+\infty} t^2\frac{|\omega^{1/2}R_{-t^2}(\omega^2)f\rangle \langle \omega^{1/2}R_{-t^2}(\omega^2)f|}{1 + 4\lambda \langle f |\omega R_{-t^2}(\omega^2) f\rangle }\d t + \omega,
\end{multline}
and 
\begin{multline}\label{equation_omega_xi_omega}
    \omega^{-1/2}\xi_{\lambda}^2 \omega^{-1/2} =  \left(\frac{8\lambda}{\pi}\right)^2 \int_0^{+\infty}\int_0^{+\infty} \frac{t^2 s^2 \langle\omega^{1/2} R_{-t^2}(\omega^2) f| \omega^{1/2}R_{-s^2}(\omega^2) f\rangle |R_{-t^2}(\omega^2) f\rangle \langle R_{-s^2}(\omega^2)f|}{\big(1 + 4\lambda \langle f | \omega R_{-t^2}(\omega^2)f\rangle \big)\big(1 + 4 \lambda \langle f |\omega R_{-s^2}(\omega^2) f\rangle\big)}\d s \d t \\ + \frac{8\lambda}{\pi}\int_0^{+\infty} t^2\frac{|\omega R_{-t^2}(\omega^2)f\rangle \langle R_{-t^2}(\omega^2)f|}{1 + 4\lambda \langle f |\omega R_{-t^2}(\omega^2) f\rangle }\d t + \frac{8\lambda}{\pi}\int_0^{+\infty} t^2\frac{|R_{-t^2}(\omega^2)f\rangle \langle \omega R_{-t^2}(\omega^2)f|}{1 + 4\lambda \langle f |\omega R_{-t^2}(\omega^2) f\rangle }\d t + \omega.
\end{multline}
We have
\begin{equation}\label{equation_inegalite_partie_3}
    \int_0^{+\infty} t^2\frac{\|\omega R_{-t^2}(\omega^2)f\| \|R_{-t^2}(\omega^2)f\|}{1 + 4\lambda \langle f | \omega R_{-t^2}(\omega^2)f\rangle} \d t \leq \int_0^{+\infty} t^2\|\omega R_{-t^2}(\omega^2)f\| \|R_{-t^2}(\omega^2)f\|\d t.
\end{equation}
We want to show that 
\begin{equation} 
	\int_0^{+\infty} t^2\|\omega R_{-t^2}(\omega^2)f\| \|R_{-t^2}(\omega^2)f\|\d t<+\infty. 
\end{equation}
As $\omega \geq 1$ and $x \in [e^4, +\infty)\mapsto \ln^4(x)/x$ is decreasing, we have for $t\in R_+$,
\begin{equation}
    \frac{\ln\big(e^4 \omega^2\big)^4}{\ln\big(e^4 \omega^2 + e^4 t^2\big)^4} \geq \frac{\omega^2}{\omega^2 + t^2}.
\end{equation}
Therefore,
\begin{equation}
    \frac{\omega^2}{\big(\omega^2 + t^2\big)^2}\leq \frac{\big(4 + 2\ln(\omega)\big)^4}{\big(\omega^2 + t^2\big)\big(4 + \ln(\omega^2 +t^2)\big)^4} \leq \frac{\big(2 + \ln(\omega)\big)^4}{\big(1 + t^2\big)\big(2 + \ln(1 +t)\big)^4},
\end{equation}
which means that 
\begin{equation}
    \big\|\omega R_{-t^2}(\omega^2)f\big\| \leq \frac{\big\| \big(2 + \ln(\omega)\big)^2f \big\|}{\sqrt{1 + t^2}\big(2 + \ln(1 + t^2)\big)^2}.
\end{equation}
Finally, we find
\begin{equation}
    \int_0^{+\infty} t^2\frac{\|\omega R_{-t^2}(\omega^2)f\| \|R_{-t^2}(\omega^2)f\|}{1 + 4\lambda \langle f | \omega R_{-t^2}(\omega^2)f\rangle} \d t \leq \big\| \big(2 + \ln(\omega)\big)^2f \big\| \|f\| \int_0^{+\infty} \frac{\d t}{\sqrt{1 + t^2}\big(2 + \ln(1 + t^2)\big)^2}.
\end{equation}
As $t\in \R_+ \mapsto \frac{1}{(1 + t)\ln(1+t)^2}$ is integrable, we have 
\begin{equation}
    \int_0^{+\infty} t^2\frac{\|\omega R_{-t^2}(\omega^2)f\| \|R_{-t^2}(\omega^2)f\|}{1 + 4\lambda \langle f | \omega R_{-t^2}(\omega^2)f\rangle} \d t  < +\infty.
\end{equation}
Besides,
\begin{equation}
    \int_0^{+\infty} t^2\frac{\|\omega^{1/2} R_{-t^2}(\omega^2)f\|^2}{1 + 4\lambda \langle f | \omega R_{-t^2}(\omega^2)f\rangle} \d t\leq \frac{\pi}{4}\|f\|^2 < + \infty.
\end{equation}
Therefore, $\xi_{\lambda}\omega^{-1}\xi_{\lambda}- \omega$ and $\omega^{-1/2} \xi_{\lambda}^2 \omega^{-1/2} - \omega$ are trace class and by taking the trace in (\ref{equation_omega_xi_omega}) and (\ref{equation_xi_omega_xi}), we find 
\begin{equation}
    \tr\big(\xi_{\lambda}\omega^{-1}\xi_{\lambda} - \omega\big) = \tr\big(\omega^{-1/2} \xi_{\lambda}^2 \omega^{-1/2} - \omega\big).
\end{equation}
\end{proof}
\begin{remark}[Assumptions on $\omega$]
    If we compute things naively in (\ref{equation_inegalite_partie_3}), we find 
    \begin{equation}
        t^2\|\omega R_{-t^2}(\omega^2)f\| \|R_{-t^2}(\omega^2)f\|\leq  t^2 \frac{1}{t}\|f\| \frac{1}{t^2}\|f\| = \frac{1}{t}\|f\|^2
    \end{equation}
    which gives a logarithmic divergence when integrated - this is why we impose a logarithmic regularity assumption on $f$. Note that on the other hand, $t^2 \|\omega^{1/2} R_{-t^2}(\omega^2)f\|^2$ is integrable for a general $f\in \Hcal$.\\
    
\end{remark}

The proof of Proposition \ref{theorem_0} is then the combination of Lemma \ref{lemme_nam_etc}, Lemma \ref{lemme_verification_hypothese_trace_class} and Lemma \ref{lemme_calcul_constante}.

\section{Existence of $\xil$}\label{section_existence_xil}
The renormalization of the one-body operator can be found in the fourth section of \cite{derezinski_renormalization_2002}. Nonetheless, we give a short presentation for completeness.
\subsection{The case $\omega^{-1/2}f\in \Hcal$} \label{ssection_existence_regulier}

In this section, we prove (\ref{equation_resolvante}) in the case $\omega^{1/2}f \in \Hcal$, and then show that the expression is well defined for less regular $f$, namely $\omega^{-1/2}f \in \Hcal$. Then, we briefly explain why this expression is actually the resolvent of an operator that we will call $\xi_{\lambda}^2$. First of all, to get (\ref{equation_resolvante}), we use the following lemma (see e.g. Theorem 1.1.1 of \cite{albeverio_singular_2000}).

\begin{lemma}[Resolvent of a rank one perturbation]\label{lemme_resolvante_rang_un}
    Let $A$ be a self-adjoint operator whith dense domain, $z \in \C \setminus \R$, $\psi \in \Hcal$ and $\alpha \in \R$. We set $T_{\alpha} = A + \alpha|\psi\rangle \langle \psi |$, then
    \begin{equation}\label{equation_resolvante_rang_un}
        R_z(T_{\alpha}) = R_z(A) - \frac{\alpha}{1 + \alpha \langle \psi | R_z(A)\psi\rangle} R_z(A) |\psi\rangle \langle \psi| R_z(A).
    \end{equation}
\end{lemma}

We apply Lemma \ref{lemme_resolvante_rang_un} to $A = \omega^2$, $\psi = \omega^{1/2}f$ and $\alpha = 4 \lambda$, which gives the expected equation 

\begin{equation}
    R_z(\lambda, f) = \big(\omega^2 - z\big)^{-1} - \frac{4\lambda }{1 + 4\lambda \langle f | \omega(\omega^2 -z)^{-1}f\rangle}\omega^{1/2}\big(\omega^2 - z)^{-1}|f\rangle\langle f|\omega^{1/2}\big(\omega^2 - z)^{-1}. \tag{\ref{equation_resolvante}}
\end{equation}
As 
\begin{equation}
	(\omega^2 - z)^{-1} = \omega^{-2} + z \omega^{-2}(\omega^2 - z)^{-1},
\end{equation} 
we see that the right-hand side is well defined when $\omega^{-1/2}f\in \Hcal$, and we still call it $R_z(\lambda, f)$. We now use the following lemma, whose proof can be found in that of Theorem 2.17 of \cite{kato_perturbation_1995}, to show that $R_z(\lambda, f)$ is indeed the resolvent of a self-adjoint operator (once again, see \cite{derezinski_renormalization_2002} for more details).

\begin{lemma}\label{lemme_operateur_resolvante}
    Let $(R(z))_{z\in \C \setminus \R}$ be a family of bounded operators on $\Hcal$ such that 
    \begin{itemize}
        \item $\forall z \in \C \setminus \R,~ R(z)^* = R(\bar{z})$;
        \item $\forall z, w \in \C \setminus \R,~ R(z+w) - R(w) = zR(z+w)R(w)$;
        \item $\cap_z \ker R(z) = \{0\}.$
    \end{itemize} 
    Then there exists a unique self-adjoint operator whose resolvent on $\C \setminus \R$ is $(R(z))_{z\in \C \setminus \R}$.
\end{lemma}

\noindent Let us check that $R_z(\lambda, f)$ satisfies the assumptions of Lemma \ref{lemme_operateur_resolvante}:
\begin{itemize}
    \item Let $z \in \C \setminus \R$, from the explicit expression of $R_z(\lambda, f)$ and the fact that $\big((\omega^2 - z)^{-1}\big)^* = (\omega^2 - \bar{z})^{-1}$, it is clear that $R_z(\lambda, f)^* = R_{\bar{z}}(\lambda, f)$.
    \item Let $z, w \in \C \setminus \R$, we know that the resolvent of $\omega^2$ satisfies 
    \begin{equation}\label{equation_egalite_resolvant_omega_carre}
        R_{z + w}(\omega^2) - R_w(\omega^2) = z R_{z + w}(\omega^2)R_w(\omega^2).
    \end{equation}
	By a direct -- though tedious -- computation, using \eqref{equation_egalite_resolvant_omega_carre}, one indeed finds the equality
    \begin{equation}
        R_{z + w}(\lambda, f) - R_{w}(\lambda, f) = zR_{z + w}(\lambda, f) R_w(\lambda, f).
    \end{equation}
    \item Let $z \in \C$, we assume that $\ker R_z(\lambda, f) \neq \{0\}$, and we take $\varphi \in \ker R_z(\lambda, f) \setminus \{0\}$. This implies that 
    \begin{equation}
    	(\omega^2 - z)^{-1}\varphi \in \Span (\omega^2 - z)^{-1}\omega^{1/2}f,
    \end{equation} 
	and as $(\omega^2 - z)^{-1}$ is injective, we know that $(\omega^2 - z)^{-1}\varphi \neq 0$ and we can simplify the expression and get
    \begin{equation}
        1 = \frac{4\lambda\langle f | \omega (\omega^2 - z)^{-1} f\rangle}{1 + 4\lambda\langle f | \omega (\omega^2 - z)^{-1} f\rangle},
    \end{equation}
    which is impossible.
\end{itemize}
Therefore, there exists a unique self-adjoint operator whose resolvent is $R_z(\lambda, f)$. As its resolvent set contains $\R_-$, this operator is positive, we call it $\xi_{\lambda}^2$, and we call $\xi_{\lambda}$ its positive square root.

\begin{remark}
    As $\xi_{\lambda}$ is now defined by the resolvent of its square, (\ref{equation_xilambda_integrale}) still holds true, which implies that $\xi_{\lambda} \geq \omega \geq 1$.
\end{remark}

\subsection{The case $\omega^{-1/2}f \notin \Hcal$ and $\omega^{-3/2}f\in \Hcal$}\label{ssection_existence_irregulier}

Here, we prove (\ref{equation_resolvante_renormalisee}) in the case $\omega^{1/2}f \in \Hcal$, and then show that the expression is well defined for less regular $f$, namely $\omega^{-3/2}f \in \Hcal$. Then, we explain why this expression is actually the resolvent of an operator that we will call $\xi_{\lambda}^2$. Let $\lambda \in \R_-$ be the renormalized coupling constant. If $\omega^{1/2}f \in \Hcal$, we set $\tilde{\lambda}^{-1} = \lambda^{-1} - 4 \|\omega^{-1/2}f\|^2 < 0$. We have
\begin{equation}
    \frac{1}{4\tilde{\lambda}} + \langle f | \omega (\omega^2 - z)^{-1}f\rangle = \frac{1}{4\lambda} -  \|\omega^{-1/2}f\|^2 + \big\langle f \big| \omega \big(\omega^{-2} + z\omega^{-2}(\omega^2 - z)^{-1}\big) f\big\rangle = \frac{1}{4\lambda} +  z\langle f |\omega^{-1}(\omega^2 - z)^{-1}f\rangle,
\end{equation}
from which we deduce
\begin{equation}
    \tilde{R}_z(\lambda, f) := R_z(\tilde{\lambda}, f) = \big(\omega^2 - z\big)^{-1} -\frac{4\lambda}{1 + 4\lambda z \langle f |\omega^{-1}(\omega^2 - z)^{-1} f\rangle}\omega^{1/2}(\omega^2 - z)^{-1}|f\rangle\langle f| \omega^{1/2}(\omega^2 - z)^{-1}.\tag{\ref{equation_resolvante_renormalisee}}
\end{equation}
The right-hand side is well defined for $\omega^{-3/2}f\in \Hcal$ and for all $\lambda \in \mathbb{R}$ and $z \in \C \setminus \R$, and we still call it $\tilde{R}_z(\lambda, f)$. Now, we briefly check the assumptions of Lemma \ref{lemme_operateur_resolvante}:
\begin{itemize}
    \item Let $z\in \C \setminus \R$, from the explicit expression of $\tilde{R}$, we find $\tilde{R}_z(\lambda, f)^* = \tilde{R}_{\bar{z}}(\lambda, f)$.
    \item Let $z, w \in \C \setminus \R$, by a direct -- though tedious -- computation using the resolvent equality \eqref{equation_egalite_resolvant_omega_carre}, we find $\tilde{R}_{z + w}(\lambda, f) - \tilde{R}_{w}(\lambda, f) = z\tilde{R}_{z + w}(\lambda, f) \tilde{R}_w(\lambda, f)$.
    \item Let $z\in \C \setminus \R$ and $\varphi \in \ker \tilde{R}_z(\lambda, f)$, then we have $(\omega^2 - z)^{-1} \varphi \in \Span \big(\omega^{1/2} (\omega^2 - z)^{-1}f\big)$. As $(\omega^2 - z)^{-1} \varphi \in D(\omega^2)$ and $\omega^{1/2} (\omega^2 - z)^{-1}f \notin D(\omega^2)$ because $f \notin D(\omega^{1/2})$, we have $(\omega^2 - z)^{-1} \varphi = 0$, and thus $\varphi = 0$.
\end{itemize}
Therefore, there exists a unique self-adjoint operator whose resolvent is $\tilde{R}_z(\lambda, f)$. As its resolvent set contains $\R_-$, this operator is positive; we call it $\xi_{\lambda}^2$, and we call $\xi_{\lambda}$ its positive square root.

\section{Proof of Theorem \ref{theorem1}}\label{section_preuve_theoreme1}

We divide the proof of the case $\omega^{-1/2}f\in \Hcal$ in three parts; first, we show in Proposition \ref{proposition_existence_trace} the existence of the formal trace defined by the integral formula (\ref{equation_trace_formelle}), then we prove in Proposition \ref{proposition_existence_transformations_bogolioubov} the existence of Bogoliubov transformations, and we finally show in Proposition \ref{proposition_convergence_diagonalisé} the convergence of the regularized model.\\

\noindent We recall that $f_n = \mathds{1}_{\omega \leq n}f$. From the definition, we get 
    \begin{equation}
        |f_n|^2 \leq n \omega^{-1}|f|^2 ~~~~~~~~\mathrm{and}~~~~~~ \omega^{2s}|f_n|^2 \leq n^{1+2s} \omega^{-1}|f|^2,
    \end{equation}
    which shows that $f_n$ is well defined in $\Hcal$ and that $\omega^{s}_nf \in \Hcal$ for all $s\in \mathbb{R}_+$. Moreover, it is clear that $\omega^{-1/2}f_n \to \omega^{-1/2}f$ in $\Hcal$ and that $|f_n|^2 \leq |f|^2$.

\subsection{Existence of the trace}\label{subsection_trace}
In this section, we show that even if $\xi_{\lambda} - h_{\lambda}$ is not necessarily trace-class, it can be approximated by a sequence of trace-class operators whose traces converge provided that $\omega^{-1/4}f\in \Hcal$.

\begin{proposition}[Existence of the formal trace]\label{proposition_existence_trace}~\\
    If $\omega^{1/2}f \in \Hcal$,  $\xi_{\lambda} - h_{\lambda}$ is trace-class, and we have
    \begin{equation}\label{equation69}
        \tr\big(\xi_{\lambda} - h_{\lambda}\big) = \frac{-32\lambda^2}{\pi} \int_0^{+\infty} \frac{t^2\|R_{-t^2}(\omega^2)\omega^{1/2}f\|^2 \big\langle f\big| R_{-t^2}(\omega^2)\omega f\big\rangle }{1 + 4\lambda\big\langle f | R_{-t^2}(\omega^2) \omega f \big\rangle}\d t.
    \end{equation}
    The right-hand side of $(\ref{equation69})$ is well defined even when we only have $\omega^{-1/4}f \in \Hcal$, and we set
    \begin{equation}
            \tr_{\infty}\big(\xi_{\lambda} - h_{\lambda}\big) := \frac{-32\lambda^2}{\pi} \int_0^{+\infty} \frac{t^2\|R_{-t^2}(\omega^2)\omega^{1/2}f\|^2 \big\langle f\big| R_{-t^2}(\omega^2)\omega f\big\rangle }{1 + 4\lambda\big\langle f | R_{-t^2}(\omega^2) \omega f \big\rangle}\d t.
        \end{equation}
    Let $\xi_{\lambda, n} = \big(\omega^2 + 4\lambda \omega^{1/2}|f_n\rangle\langle f_n | \omega^{1/2}\big)^{1/2}$ and $h_{\lambda,n} = \omega + 2\lambda |f_n\rangle\langle f_n|$, then the trace of $\xi_{\lambda, n} - h_{\lambda, n}$ converges
    \begin{equation}\label{equation_convergence_trace}
        \tr\big(\xi_{\lambda,n} - h_{\lambda, n}\big) \to \tr_{\infty}\big(\xi_{\lambda} - h_{\lambda}\big).
    \end{equation}
\end{proposition}

\begin{proof} First, we prove (\ref{equation69}). Assuming $\omega^{1/2}f \ \Hcal$, and applying (\ref{equation_trace_racine_carree_perturbation}) of Lemma \ref{lemme_formules_integrales_racine_perturbation} to $A = \omega^2$, $\psi = \omega^{1/2}f$ and $\alpha = 4\lambda$, we get
    \begin{equation}
        \tr\big(\xi_{\lambda} - \omega\big) = \frac{8\lambda}{\pi}\int_0^{+\infty} \frac{t^2\big\|  R_{-t^2}(\omega^2)\omega^{1/2}f\big\|^2}{1 + 4\lambda\big\langle f |  R_{-t^2}(\omega^2)\omega f\big\rangle} \d t.
    \end{equation}
    As $\int_0^{+\infty} \frac{at^2}{(a^2 + t^2)^2}\d t = \frac{\pi}{4}$ for all $a\in \mathbb{R}_+$, we have
    \begin{equation}
        \frac{8\lambda}{\pi}\int_0^{+\infty} t^2\big\|  R_{-t^2}(\omega^2)\omega^{1/2}f\big\|^2 \d t = \Big\langle f \Big| \frac{8\lambda}{\pi}\int_0^{+\infty} t^2  R_{-t^2}(\omega^2)^2\omega \d t \big|f\Big\rangle = 2\lambda\|f\|^2 = 2\lambda \tr\big(|f\rangle\langle f|\big).
    \end{equation}
    We recall that $\hl = \omega + 2\lambda |f\rangle\langle f|$ and deduce that 
    \begin{equation}\label{equation74}
        \tr\big(\xi_{\lambda} - h_{\lambda}\big) = \frac{-32\lambda^2}{\pi} \int_0^{+\infty} \frac{t^2\|R_{-t^2}(\omega^2)\omega^{1/2}f\|^2 \big\langle f\big| R_{-t^2}(\omega^2) \omega f\big\rangle }{1 + 4\lambda\big\langle f | R_{-t^2}(\omega^2) \omega f \big\rangle}\d t.
    \end{equation}

    If $\omega^{1/2}f \notin \Hcal$ but $\omega^{-1/4}f\in \Hcal$, the integral in (\ref{equation74}) is still well defined. Indeed, the integrand is positive, and as $\int_0^{+\infty} \frac{a}{\sqrt{t}(a^2 +t^2)}\d t = \frac{\pi}{\sqrt{2}}a^{-1/2}$,
    \begin{align}
        \int_0^{+\infty} \frac{t^2\|R_{-t^2}(\omega^2)\omega^{1/2}f\|^2 \big\langle f\big| R_{-t^2}(\omega^2) \omega f\big\rangle }{1 + 4\lambda\big\langle f | R_{-t^2}(\omega^2) \omega f \big\rangle}\d t &\leq \int_0^{+\infty} t^2\|R_{-t^2}(\omega^2)\omega^{1/2}f\|^2 \big\langle f\big| R_{-t^2}(\omega^2) \omega f\big\rangle \d t \notag\\
        &\leq \int_0^{+\infty} t^{-1/2}\|\omega^{-1/4}f\|^2 \big\langle f\big| R_{-t^2}(\omega^2) \omega f\big\rangle \d t \notag\\
        &= \|\omega^{-1/4}f\|^2 \Big\langle f\Big| \int_0^{+\infty} t^{-1/2} R_{-t^2}(\omega^2)\d t \omega \Big|f\Big\rangle \notag\\
        &= \|\omega^{-1/4}f\|^2 \Big\langle f\Big| \frac{\pi}{\sqrt{2}} \omega^{-1/2} \Big|f\Big\rangle 
        = \frac{\pi}{\sqrt{2}}\|\omega^{-1/4}f\|^4.
    \end{align}

    Finally, we prove the convergence (\ref{equation_convergence_trace}). As $\omega \geq 1$ and $|f_n|^2 \nearrow |f|^2$, we have
    \begin{equation}
        \|R_{-t^2}(\omega^2)\omega^{1/2}f_n\|^2 \to \|R_{-t^2}(\omega^2)\omega^{1/2}f\|^2 ~~\mathrm{and}~~\big\langle f_n\big| R_{-t^2}(\omega^2) \omega f_n\big\rangle \to \big\langle f\big| R_{-t^2}(\omega^2) \omega f\big\rangle.
    \end{equation}
    Moreover, we have 
    \begin{equation}
        \frac{t^2\|R_{-t^2}(\omega^2)\omega^{1/2}f_n\|^2 \big\langle f_n\big| R_{-t^2}(\omega^2) \omega f_n\big\rangle }{1 + 4\lambda\big\langle f_n | R_{-t^2}(\omega^2) \omega f_n\big\rangle} \leq t^2\|R_{-t^2}(\omega^2)\omega^{1/2}f\|^2 \big\langle f\big| R_{-t^2}(\omega^2) \omega f\big\rangle.
    \end{equation}
    As the right-hand side is integrable, by Lebesgue's dominated convergence theorem, we get the expected convergence (\ref{equation_convergence_trace}), which concludes the proof of Proposition \ref{proposition_existence_trace}.
\end{proof}
\subsection{Existence of Bogoliubov transformations}
Here, we prove the existence of Bogoliubov transformations that allow us to diagonalize our Hamiltonian.
\begin{proposition}[Existence of Bogoliubov transformations]\label{proposition_existence_transformations_bogolioubov}
    Assume $\omega^{-1/2}f\in \Hcal$. We define $\xi_{\lambda}$ as in Section \ref{ssection_existence_regulier}, $U_{\lambda}$ and $V_{\lambda}$ by $(\ref{equation_definition_lemme_diagonalisation})$, and $\Vcal_{\lambda}$ by $(\ref{equation_definition_lemme_diagonalisation2})$. Then, $\Vcal_{\lambda} \in \Gcal$ i.e. $U_{\lambda}$ and $V_{\lambda}$ are implemented by a Bogoliubov transformation.
\end{proposition}
\begin{proof}
It is clear that $\Vcal_{\lambda}^* \Scal \Vcal_{\lambda} = \Vcal_{\lambda} \Scal \Vcal_{\lambda}^* = \Scal$. Thanks to Proposition \ref{proposition_existence_transformation_bogolioubov}, this implies that there exist Bogoliubov transformations associated with $\Vcal_{\lambda}$ if and only if $\tr(V_{\lambda}V_{\lambda}^*) < + \infty$. \\
As already used in Section \ref{section_preuve_proposition}, we have 
\begin{equation}
    4V_{\lambda} V_{\lambda}^* = \omega^{-1/2} \xi_{\lambda} \omega^{-1/2} - 1 + \omega^{1/2} \xi_{\lambda}^{-1} \omega^{1/2} - 1.\tag{\ref{equation_VVetoile}}
\end{equation}
Moreover, as $\xi_{\lambda}$ is now defined by the resolvent of its square, we can still use Lemma \ref{lemme_formules_integrales_racine_perturbation}. Thus, we find, as in $(\ref{equation_xilambda_integrale})$,
\begin{equation}\label{equation_valeur_integrale_xilambda}
    \omega^{-1/2}\xi_{\lambda} \omega^{-1/2} - 1 = \frac{8\lambda}{\pi}\int_0^{+\infty}\frac{t^2 |R_{-t^2}(\omega^2)f\rangle \langle R_{-t^2}(\omega^2)f|}{1 + 4\lambda \langle f | \omega R_{-t^2 }(\omega^2) f\rangle} \d t,
\end{equation}
which is still a self-adjoint positive operator. As it is positive, we can compute the trace a priori. Let $(u_j)$ be an orthonormal basis of $\Hcal$, we have 
\begin{align}\label{equation_trace_bornee_1}
    \sum_{j = 1 }^{+\infty} \langle u_j| \big(\omega^{-1/2}\xi_{\lambda}\omega^{-1/2} - 1 \big) |u_j\rangle &= \frac{8\lambda}{\pi}\int_0^{+\infty} \frac{t^2}{1 + 4\lambda \langle f | \omega R_{-t^2 }(\omega^2) f\rangle}\sum_{j=1}^{+\infty} \big|\langle u_j | R_{-t^2 }(\omega^2) f\rangle \big|^2 \d t \notag\\
    &= \frac{8\lambda}{\pi}\int_0^{+\infty} \frac{t^2 \|R_{-t^2}(\omega^2) f\|^2}{1 + 4\lambda \langle f | \omega R_{-t^2 }(\omega^2) f\rangle}\d t\notag\\
    &\leq \frac{8\lambda}{\pi}\int_0^{+\infty} t^2 \|R_{-t^2}(\omega^2) f\|^2\d t = \frac{8\lambda}{\pi} \Big\langle f\Big|\frac{\pi}{4}\omega^{-1}f\Big\rangle = 2\lambda\|\omega^{-1/2}f\|^2.
\end{align}
On the other hand, we have
\begin{equation}\label{equation_valeur_integrale_xilambda_inverse}
    \omega^{1/2} \xi_{\lambda}^{-1} \omega^{1/2} - 1 = \frac{-8\lambda}{\pi} \int_0^{+\infty} \frac{|R_{-t^2}(\omega^2)\omega f\rangle \langle R_{-t^2}(\omega^2) \omega f|}{1 + 4 \lambda \langle f |\omega R_{-t^2}(\omega^2) f\rangle}\d t
\end{equation}
which is a self-adjoint negative operator, and 
\begin{equation}\label{equation_trace_bornee_2}
    \sum_{j = 1}^{+\infty}\langle u_j | \big(1 - \omega^{1/2} \xi_{\lambda}^{-1} \omega^{1/2}\big)|u_j\rangle \leq \frac{8\lambda}{\pi} \int_0^{+\infty} \|R_{-t^2}(\omega^2)\omega f\|^2 \d t = \frac{8\lambda}{\pi}\Big\langle f \Big| \frac{\pi}{4}\omega^{-1} \Big|f\Big\rangle = 2\lambda \|\omega^{-1/2}f\|^2.
\end{equation}
From (\ref{equation_VVetoile}), (\ref{equation_trace_bornee_1}) and (\ref{equation_trace_bornee_2}), we can then deduce that $V_{\lambda} V_{\lambda}^*$ is indeed trace-class, which means that we have indeed $\Vcal_{\lambda} \in \Gcal$.
\end{proof}

\subsection{Convergence of the regularized Hamiltonian}
Now that we have discussed the convergence of the traces and the existence of Bogoliubov transformations in the limit, we can prove the convergence of the regularized Hamiltonian.

\begin{proposition}[Convergence of the regularized Hamiltonian] \label{proposition_convergence_diagonalisé}
    Assume $\omega^{-1/2}f\in \Hcal$. We can choose a Bogoliubov transformation $\mathbb{U}_{\lambda, n}$ such that we have the following convergence 
    \begin{equation}
        \mathbb{U}_{\lambda, n}\dg(\xi_{\lambda, n})\mathbb{U}_{\lambda,n}^* \to \mathbb{U}_{\lambda}\dg(\xi_{\lambda})\mathbb{U}_{\lambda}^*
    \end{equation}
    in the strong resolvent sense.
\end{proposition}
To prove Proposition \ref{proposition_convergence_diagonalisé}, we need four steps. First, we show in Lemma \ref{lemme_convergence_norme_seconde_quantification} that the convergence of $\dg(\xi_{\lambda, n})$ is implied by that of $\xi_{\lambda, n}$ and in Lemma \ref{lemme_convergence_xil} we show that $\xi_{\lambda, n}$ converges. Then, we prove in Lemma \ref{lemme_convergence_transformations} that the unitary operators $\mathbb{U}_{\lambda, n}$ converge under certain assumptions that we check in Lemma \ref{lemme_convergence_VU}.

\begin{lemma}[Convergence of a second quantization]\label{lemme_convergence_norme_seconde_quantification}
    We set $\varepsilon > 0$. If $(A_n)$ is a sequence of self-adjoint operators acting on $\Hcal$, satisfying $(-\infty, \varepsilon) \subset \rho(A_n)$ for all $n\in \N$ and converging to the self-adjoint operator $A$ in the norm resolvent sense, then $\dg(A_n)$ converges to $\dg(A)$ in the norm resolvent sense.
\end{lemma}
\begin{proof}
    Because of Theorem VIII.19 of \cite{reed_i_1980}, we only have to prove that $\dg(A_n)^{-1}$ converges in norm to $\dg(A)^{-1}$. For $\Psi \in \Hcal^{\otimes_{s} N}\subset \Fcal$, we write $\dg(A_n)\Psi = \sum_{j=1}^N A_{n, j} \Psi$ for $n \in \N \cup\{+\infty\}$ with the convention $A_{\infty} = A$. Then, 
    \begin{equation}
        \langle \Psi| \big(\dg(A_n)^{-1} - \dg(A)^{-1}\big)\Psi\rangle = \langle \Psi| \dg(A_n)^{-1} \big(\dg(A) - \dg(A_n)\big)\dg(A)^{-1}\Psi\rangle.
    \end{equation}
    Let us write $D_{n, j} = A_{n,j}^{-1} - A_{\infty, j}^{-1}$ and $D_n = A_n^{-1} - A^{-1}$. We find
    \begin{equation}
        \langle \Psi| \big(\dg(A_n)^{-1} - \dg(A)^{-1}\big)\Psi\rangle = \sum_{j = 1}^N \langle \Psi |\dg(A_n)^{-1}A_{n, j}D_{n, j}A_{\infty, j}\dg(A)^{-1}\Psi\rangle.
    \end{equation}
    We can then use the Cauchy-Schwarz inequality to find
    \begin{equation}
        \big|\langle \Psi| \big(\dg(A_n)^{-1} - \dg(A)^{-1}\big)\Psi\rangle\big| \leq \|D_n\| \sqrt{\sum_{j= 1}^N\|A_{n, j}\dg(A_n)^{-1}\Psi\|^2} \sqrt{\sum_{j=1}^N\|A_{\infty, j}\dg(A)^{-1}\Psi\|^2}.
    \end{equation}
    As the $A_{n, j}$ are positive operators, we know that $\big(\sum_j A_{n, j}\big)^2 \geq \sum_j A_{n, j}^2$. In particular, this implies that 
    \begin{equation}
        \sum_{j= 1}^N\|A_{n, j}\dg(A_n)^{-1}\Psi\|^2 = \Big\langle \Psi\Big| \dg(A_n)^{-1} \sum_{j = 1}^N A_j^2 \dg(A_n)^{-1}\Psi\Big\rangle \leq \|\Psi\|^2.
    \end{equation}
    Therefore, we have
    \begin{equation}
        \big|\langle \Psi| \big(\dg(A_n)^{-1} - \dg(A)^{-1}\big)\Psi\rangle\big| \leq \|D_n\| \|\Psi\|^2.
    \end{equation}
    It is clear that the inequality still holds for a general $\Psi \in \Fcal$, and thus $\dg(A_n)^{-1}$ converges in norm to $\dg(A)^{-1}$.
\end{proof}
We want to apply Lemma \ref{lemme_convergence_norme_seconde_quantification} to $\xi_{\lambda, n}$, and thus have to prove the convergence of $\xi_{\lambda, n}$.

\begin{lemma}[Convergence of $\xi_{\lambda, n}$]\label{lemme_convergence_xil}
The operators $\xi_{\lambda, n}$ converge to $\xi_{\lambda}$ in the norm resolvent sense, i.e.
\begin{equation}
    \forall z \in \C \setminus [1, +\infty), \|R_z(\xi_{\lambda, n}) - R_z(\xi_{\lambda})\| \to 0.
\end{equation}
\end{lemma}

Note that the norm resolvent convergence of $\xi_{\lambda, n}$ implies its strong resolvent convergence.

\begin{proof}
First, we show that $\xi_{\lambda, n}^2$ converges to $\xi_{\lambda}^2$ in the norm resolvent sense, then we deduce that $\xi_{\lambda, n}^{-1}$ converges in norm to $\xi_{\lambda}^{-1}$ which allows us to conclude.\\

	For $n\in \N \cup \{+\infty\}$, we set 
    \begin{equation}
        B_{n} = \frac{4\lambda}{1 + 4\lambda \langle f_{n} |\omega(\omega^2 - z)^{-1}f_{n}\rangle}\in \R_+,~~~~\text{and}~~~~ \psi_n = \omega^{1/2}(\omega^2 - z)^{-1}f_n \in \Hcal,
        \end{equation}
        with convention $f_{\infty} = f$, $B_{\infty} = B$ and $\psi_{\infty} = \psi$.
        For $\varphi \in \Hcal$,
        \begin{align}
            \big(R_z(\lambda, f_n) - R_z(\lambda, f)\big)\varphi &= B_n |\psi_n\rangle\langle \psi_n | \varphi \rangle - B |\psi\rangle\langle \psi | \varphi \rangle \notag \\ &= (B_n - B)\langle \psi | \varphi\rangle \psi + B_n \langle \psi_n -\psi|\varphi\rangle \psi + B_n\langle \psi_n|\varphi\rangle (\psi_n -\psi),
        \end{align}
        and then
        \begin{equation}
            \big\|\big(R_z(\lambda, f_n) - R_z(\lambda, f)\big)\varphi\big\| \leq \Big(\|\psi\|^2|B_n - B| + |B_n|\|\psi\|\|\psi_n - \psi\| + |B_n|\|\psi_n\|\|\psi_n - \psi\|\Big)\|\varphi\|.
        \end{equation}
        As $\omega^{-1/2}f_n \to \omega^{-1/2}f$, we have $B_n \to B$ and $\psi_n \to \psi$. This implies that
        \begin{equation}
            \|\psi\|^2|B_n - B| + |B_n|\|\psi\|\|\psi_n - \psi\| + |B_n|\|\psi_n\|\|\psi_n - \psi\| \to 0,
        \end{equation}
        and therefore $R_{z}(\xi_{\lambda, n}^2) \to R_z(\xi_{\lambda}^2)$ in the sense of the norm. Note that this convergence is uniform in $z$ for $z\in K$ if $K$ is a compact subset of $\C \setminus [1, +\infty)$.\\
        
    Moreover,
    \begin{equation}
        \xi_{\lambda,n}^{-1} = \frac{2}{\pi} \int_0^{+\infty}R_{-t^2}(\xi_{\lambda,n}^2)\d t.
    \end{equation}
    Let $\eta > 0$ and $\varphi \in \Hcal$, there exists $\alpha > 0$ such that $\int_{\alpha}^{+\infty} t^{-2} \d t < \eta$. Then, 
    \begin{equation}
        \|\xi_{\lambda, n}^{-1} \varphi - \xi_{\lambda}^{-1}\varphi\| \leq \int_0^{\alpha}\|R_{-t^2} (\xi_{\lambda, n}^2)\varphi - R_{-t^2}(\xi_{\lambda}^2)\varphi\| \d t + \eta \|\varphi\|.
    \end{equation}
    As $R_{-t^2} (\xi_{\lambda, n}^2)$ converges in norm to $R_{-t^2}(\xi_{\lambda}^2)$, uniformly in $t \in [0,\alpha]$ we finally have
    \begin{equation}
        \|\xi_{\lambda, n}^{-1} - \xi_{\lambda}^{-1}\| \to 0~~\text{ when }~~ n \to + \infty.
    \end{equation}
    Using Theorem VIII.19 of \cite{reed_i_1980}, this concludes the proof of Lemma \ref{lemme_convergence_xil}.

\end{proof}

Now, we want to show that under certain assumptions, the convergence of $(U_{\lambda, n})$ and $(V_{\lambda, n})$ implies the convergence of a sequence of associated Bogoliubov transformations.
\begin{lemma}[Convergence of Bogoliubov transformations]\label{lemme_convergence_transformations}
    Let $(\mathbb{U}_{\lambda, n})$ be a sequence of Bogoliubov transformations associated with $(U_{\lambda, n})$ and $(V_{\lambda, n})$ -- see \eqref{equation_definition_transformation_bogoliubov}. We assume that $(U_{\lambda,n})$, $(U_{\lambda, n}^*)$, $(V_{\lambda, n})$ and $(V_{\lambda, n}^*)$ converge respectively to $U_{\lambda}$, $U_{\lambda}^*$, $V_{\lambda}$ and $V_{\lambda}^*$ strongly on $\Hcal$, $U_{\lambda}$ and $V_{\lambda}$ being associated with a Bogoliubov transformation $\mathbb{U}_{\lambda}$. Moreover, we assume that $(\tr(V_{\lambda, n}^*V_{\lambda, n}))$ is bounded. Then, up to a change of phase in the $\mathbb{U}_{\lambda, n}$,
    \begin{equation}
        \mathbb{U}_{\lambda, n} \to \mathbb{U}_{\lambda},
    \end{equation}
    strongly on $\Fcal$.
\end{lemma}
\begin{proof} We recall that $\Omega$ is the vacuum and that there is no canonical choice of Bogoliubov transformation associated to operators $U$ and $V$. Therefore we can choose the transformations to have convergence. First, we set a choice of $\mathbb{U}_{\lambda}$. Then, for each $n \in \N$, we choose $\mathbb{U}_{\lambda, n}$ such that 
\begin{equation}
    \langle \mathbb{U}_{\lambda}\Omega | \mathbb{U}_{\lambda, n}\Omega\rangle \in \R_+.
\end{equation}
We prove that for all $\Psi \in \Fcal$, $\mathbb{U}_{\lambda, n}\Psi \to \mathbb{U}_{\lambda}\Psi$, and we divide this proof in three parts: first, we treat the case $\Psi = \Omega$, then we prove the convergence for a finite number of particles, and we conclude by linearity and density arguments.\\
    \textbf{The case of the vacuum.} We start by showing that $\mathbb{U}_{\lambda, n}\Omega \to \mathbb{U}_{\lambda}\Omega$.\\

        The sequence $(\mathbb{U}_{\lambda, n}\Omega)$ is bounded in $\Fcal$ as $\|\mathbb{U}_{\lambda, n}\Omega\| = \|\Omega\| = 1$, therefore, there exists $\Psi \in \Fcal$ such that up to a subsequence 
        \begin{equation}
            \mathbb{U}_{\lambda, n}\Omega \rightharpoonup \Psi.
        \end{equation}
        Moreover, the sequence $\|\Ncal^{1/2}\mathbb{U}_{\lambda, n}\Omega\|$ is bounded. Indeed, fixing $(u_j)$ an orthonormal basis of $\Hcal$, we have
        \begin{equation}
            \|\Ncal^{1/2}\mathbb{U}_{\lambda, n}\Omega\|^2 = \sum_{j\in \N}\langle \Omega|\mathbb{U}_{\lambda, n}^* a^*(u_j) a(u_j) \mathbb{U}_{\lambda, n} \Omega\rangle = \sum_{j\in \N} \langle \Omega| a(V_{\lambda, n}Ju_j) a^*(V_{\lambda, n}J u_j)\Omega\rangle = \tr(V_{\lambda, n}^*V_{\lambda, n}),
        \end{equation}
        which is bounded uniformly in $n$ by assumption. Therefore, there exists $\Psi_2 \in \Fcal$ such that up to a subsquence, we have
        \begin{equation}
            (1+\Ncal^{1/2}) \mathbb{U}_{\lambda, n} \Omega \rightharpoonup \Psi_2.
        \end{equation}
        As $\|(1 + \Ncal^{1/2}) \Phi\| \geq \|\Phi\|$ for all $\Phi$, we find $\Psi_2 = (1 + \Ncal^{1/2})\Psi$ and therefore
        \begin{equation}
            (1 + \Ncal^{1/2})\mathbb{U}_{\lambda, n}\Omega \rightharpoonup (1 + \Ncal^{1/2})\Psi.
        \end{equation}
        
        We want to show that $\Psi \in \Span \mathbb{U}_{\lambda}\Omega$. Note that it is true if $\Psi \in \ker a(g)\mathbb{U}_{\lambda}^*$ for all $g\in \Hcal$. Let $g \in D(\omega)$, we have for $n\in \N$,
        \begin{align}
            \mathbb{U}_{\lambda}a(g)\mathbb{U}_{\lambda}^*\mathbb{U}_{\lambda, n} \Omega &= \big(a(U_{\lambda}^*g) - a^*(J^*V_{\lambda}^* g)\big) \mathbb{U}_{\lambda, n}\Omega \notag\\
            &= \Big(a\big((U_{\lambda}^* - U_{\lambda, n}^*)g\big) - a^*\big(J^*(V_{\lambda}^* - V_{\lambda, n}^*)g\big)\Big)\mathbb{U}_{\lambda, n}\Omega,
        \end{align}
        because $\big(a(U_{\lambda, n}^*g) - a^*(J^*V_{\lambda, n}^* g)\big) \mathbb{U}_{\lambda, n}\Omega = \mathbb{U}_{\lambda, n} a(g)\Omega = 0$.
        We know (consequence of (\ref{equation_inegalite_creation_annihilation}), see also Theorem 3.51 of \cite{derezinski_mathematics_2013}) that for $\phi \in \Hcal$ and $\Phi \in \Fcal$, we have the following inequalities
            \begin{equation}
                \|a^*(\phi)\Phi\|^2 \leq \|\phi\|^2 \|(\Ncal + 1)^{1/2}\Phi\|^2~~~~~~\text{and}~~~~~~ \|a(\phi)\Phi\|^2 \leq \|\phi\|^2 \|\Ncal^{1/2}\Phi\|^2.
            \end{equation}
        Therefore, we have 
        \begin{equation}
            \|\mathbb{U}_{\lambda}a(g)\mathbb{U}_{\lambda}^*\mathbb{U}_{\lambda, n} \Omega\| \leq \big\|(U_{\lambda}^* - U_{\lambda, n}^*)g\big\| \big\|(1 + \Ncal)^{1/2} \mathbb{U}_{\lambda, n}\Omega\big\| +\big\|(V_{\lambda}^* - V_{\lambda, n}^*)g\big\| \big\|(1 + \Ncal)^{1/2} \mathbb{U}_{\lambda, n}\Omega\big\|.
        \end{equation}
        The convergence of $U_{\lambda, n}^*$ and $V_{\lambda, n}^*$ implies that $(U_{\lambda}^* - U_{\lambda, n}^*)g \to 0$ and $(V_{\lambda}^* - V_{\lambda, n}^*)g \to 0$, and we have shown that $\big((1 + \Ncal)^{1/2} \mathbb{U}_{\lambda, n}\Omega\big)$ is bounded, therefore we have 
        \begin{equation}
            a(g)\mathbb{U}_{\lambda}^*\mathbb{U}_{\lambda, n} \Omega \to 0.
        \end{equation}
        Let $\Phi \in D(\Ncal^{1/2})$, we have $\langle \Phi | a(g) \mathbb{U}_{\lambda}^*\mathbb{U}_{\lambda, n}\Omega\rangle \to 0$ and
        \begin{equation}
            \langle \Phi | a(g) \mathbb{U}_{\lambda}^*\mathbb{U}_{\lambda, n}\Omega\rangle = \langle  \mathbb{U}_{\lambda}a^*(g)\Phi | \mathbb{U}_{\lambda, n}\Omega\rangle \to \langle  \mathbb{U}_{\lambda}a^*(g)\Phi | \Psi \rangle = \langle \Phi | a(g)\mathbb{U}_{\lambda}^*\Psi \rangle,
        \end{equation}
        which implies 
        \begin{equation}
             \langle \Phi | a(g)\mathbb{U}_{\lambda}^*\Psi \rangle = 0.
        \end{equation}
        As $D(\Ncal^{1/2})$ is dense in $\Fcal$, we finally find $a(g)\mathbb{U}_{\lambda}^*\Psi = 0$, and therefore $\Psi \in \Span \mathbb{U}_{\lambda}\Omega$ as $D(\omega)$ is dense in $\Hcal$.\\
        
        Let us write $\Psi = r e^{i\theta} \mathbb{U}_{\lambda}\Omega$. We set $P_{n} = |\mathbb{U}_{\lambda, n}\Omega\rangle\langle \mathbb{U}_{\lambda, n}\Omega|$ and $P = |\mathbb{U}_{\lambda}\Omega\rangle\langle \mathbb{U}_{\lambda}\Omega|$. For $\psi_1, \psi_2 \in \Hcal$, we have
        \begin{equation}
            \langle \psi_1|P_n\psi_2\rangle = \langle \Omega | \mathbb{U}_{\lambda, n}^* a^*(\psi_1) \mathbb{U}_{\lambda, n} \mathbb{U}_{\lambda, n}^* a(\psi_2) \mathbb{U}_{\lambda, n}\Omega\rangle = \langle V_{\lambda, n} J \psi_1 | V_{\lambda, n} J \psi_2\rangle.        \end{equation}
            The convergence of $V_{\lambda, n}$ implies that $\langle \psi_1|P_n \psi_2\rangle \to \langle \psi_1|P \psi_2\rangle$. In particular, $\langle \psi_1|P_n \psi_1\rangle \to \langle \psi_1|P \psi_1\rangle$ i.e.
            \begin{equation}
                \big|\langle \psi_1| \mathbb{U}_{\lambda, n}\Omega\rangle\big|^2 \to \big|\langle \psi_1| \mathbb{U}_{\lambda}\Omega\rangle\big|^2.
            \end{equation}
            As $\mathbb{U}_{\lambda, n}\Omega \rightharpoonup r e^{i\theta}\mathbb{U}_{\lambda}\Omega$, we also have $\big|\langle \psi_1| \mathbb{U}_{\lambda, n}\Omega\rangle\big|^2 \to r^2 \big|\langle \psi_1| \mathbb{U}_{\lambda}\Omega\rangle\big|^2$ and thus $r = 1$ and $\Psi = e^{i\theta}\mathbb{U}_{\lambda}\Omega$. 
            As $\mathbb{U}_{\lambda, n}\Omega \rightharpoonup e^{i\theta}\mathbb{U}_{\lambda}\Omega$, we have 
            \begin{equation}
                \langle \mathbb{U}_{\lambda}\Omega | \mathbb{U}_{\lambda, n}\Omega\rangle \to e^{i\theta}.
            \end{equation}
            Recall that for all $n \in \N$, $\langle \mathbb{U}_{\lambda}\Omega | \mathbb{U}_{\lambda, n}\Omega\rangle \in \R_+$, which implies that $e^{i\theta} = 1$ and $\Psi = \mathbb{U}_{\lambda}\Omega$.
            As this is the only possible limit, we have actually weak convergence of the whole sequence. As the limit has same norm as the sequence, we have convergence in norm:
            \begin{equation}
                \mathbb{U}_{\lambda, n}\Omega \to \mathbb{U}_{\lambda}\Omega.
            \end{equation}
    \textbf{The case of a state with finitely many particles.} Let $\phi_1,..., \phi_n \in \Hcal$. We set $\Phi = a^*(\phi_1)...a^*(\phi_k)\Omega$ and we want to show that $\mathbb{U}_{\lambda, n} \Phi \to \mathbb{U}_{\lambda}\Phi$. We have 
        \begin{align}
            \mathbb{U}_{\lambda, n}\Phi - \mathbb{U}_{\lambda}\Phi &= \mathbb{U}_{\lambda}a^*(\phi_1)\mathbb{U}_{\lambda}^*...\mathbb{U}_{\lambda}a^*(\phi_k)\mathbb{U}_{\lambda}^*\big(\mathbb{U}_{\lambda, n}\Omega - \mathbb{U}_{\lambda} \Omega\big) \big\}=:\Delta_0\notag\\
            & + \big(\mathbb{U}_{\lambda, n}a^*(\phi_1)\mathbb{U}_{\lambda, n}^* - \mathbb{U}_{\lambda}a^*(\phi_1)\mathbb{U}_{\lambda}^*\big)\mathbb{U}_{\lambda}a^*(\phi_2)\mathbb{U}_{\lambda}^*...\mathbb{U}_{\lambda}a^*(\phi_k)\mathbb{U}_{\lambda}^*\mathbb{U}_{\lambda, n}\Omega\big\}=:\Delta_1\notag\\
            & + ...\notag\\
            & + \mathbb{U}_{\lambda, n}a^*(\phi_1)\mathbb{U}_{\lambda, n}^*...\mathbb{U}_{\lambda, n}a^*(\phi_{k-1})\mathbb{U}_{\lambda, n}^*\big(\mathbb{U}_{\lambda, n}a^*(\phi_k)\mathbb{U}_{\lambda, n}^* - \mathbb{U}_{\lambda}a^*(\phi_k)\mathbb{U}_{\lambda}^*\big) \mathbb{U}_{\lambda, n}\Omega\big\}=:\Delta_k.
        \end{align} 
        Because of Lemma 4.4 of \cite{beyond_bogoliubov}, we have for all $n, j \in \N$,
        \begin{equation}
            \langle \mathbb{U}_{\lambda, n}\Omega | (\Ncal + 1)^j \mathbb{U}_{\lambda, n}\Omega\rangle \leq C_{n, j}
        \end{equation}
        with $C_{n, j}$ depending only on $j$, $\tr(V_{\lambda, n}^* V_{\lambda, n})$ and $\|U_{\lambda, n}\|$, which are bounded. Indeed, the boundedness of $\tr(V_n^* V_n)$ is an assumption of the Lemma, and this implies that $\|V_{\lambda, n}\|$ is bounded. As $U_{\lambda, n}^* U_{\lambda, n} = 1 + J^*V_{\lambda, n}^*V_{\lambda, n}J$, $\|U_{\lambda, n}\|$ is bounded too. This implies in particular that for all $j\in \N$, $\langle \mathbb{U}_{\lambda, n}\Omega | (\Ncal + 1)^j \mathbb{U}_{\lambda, n}\Omega\rangle$ is bounded, uniformly in $n$. Moreover, for a fixed $j$, the convergence in norm of $\mathbb{U}_{\lambda, n}\Omega$ together with the boundedness of $\mathcal{N}^{j+1}\mathbb{U}_{\lambda, n}$ gives the convergence in norm of $\mathcal{N}^j\mathbb{U}_{\lambda, n}\Omega$.\\
        
        For all $\Psi \in \Fcal$, $\psi \in \Hcal$, $j, l \in \mathbb{N}$, we have the following inequality which is a generalization of Theorem 3.51 of \cite{derezinski_mathematics_2013}:
        \begin{equation}
            \big\|(\Ncal + l)^{j/2}a^{*}(\psi)\Psi\big\|, \big\|(\Ncal + l)^{j/2}a(\psi)\Psi\big\| \leq \|\psi\| \|(\Ncal + l + 1)^{(j+1)/2}\Psi\|.
        \end{equation}
        From this, we can deduce that 
        \begin{equation}
            \|\Delta_0\| \leq \big(\|U_{\lambda}^*\phi_1\| + \|V_{\lambda}^*\phi_1\|\big)...\big(\|U_{\lambda}^*\phi_k\| + \|V_{\lambda}^*\phi_k\|\big)\big\|(\Ncal + k)^{k/2} (\mathbb{U}_{\lambda, n}\Omega - \mathbb{U}_{\lambda}\Omega)\big\|
        \end{equation}
        and for $1 \leq j \leq k$,
        \begin{multline}
            \|\Delta_j\| \leq \Big(\prod_{l_1 < j}\big(\|U_{\lambda, n}^*\phi_{l_1}\| + \|V_{\lambda, n}^*\phi_{l_1}\|\big)\Big)\big(\|(U_{\lambda, n}^* - U_{\lambda}^*)\phi_j\| + \|(V_{\lambda, n}^* - V_{\lambda}^*)\phi_j\|\big)\\
            \Big(\prod_{l_2 > j}\big(\|U_{\lambda}^*\phi_{l_2}\| + \|V_{\lambda}^*\phi_{l_2}\|\big)\Big)\big\|(\Ncal + k)^{k/2} \mathbb{U}_{\lambda, n}\Omega\big\|.
        \end{multline}
        Therefore, we can deduce that all the $\Delta_j$ converge to 0, and hence that 
        \begin{equation}
            \mathbb{U}_{\lambda, n}\Phi \to \mathbb{U}_{\lambda}\Phi.
        \end{equation}
    \textbf{Conclusion} By linearity and unitarity of the $\mathbb{U}_{\lambda, n}$ we finally have strong convergence on $\Fcal$.
\end{proof}
Now that we have proved Lemma \ref{lemme_convergence_transformations}, we need to check that its assumptions are satisfied in our case.
\begin{lemma}[Hypotheses of Lemma \ref{lemme_convergence_transformations}]\label{lemme_convergence_VU}
    Assume $\omega^{-1/2}f\in \Hcal$. We define $\xi_{\lambda}$ as in Section \ref{ssection_existence_regulier}, $U_{\lambda}$ and $V_{\lambda}$ by $(\ref{equation_definition_lemme_diagonalisation})$. We set $\xi_{\lambda, n} = \big(\omega^2 + 4 \lambda \omega^{1/2} |f_n\rangle \langle f_n |\omega^{1/2}\big)^{1/2}$ with $f_n = \mathds{1}_{\omega \leq n} f$, and we define  $U_{\lambda, n}$ and $V_{\lambda, n}$ by $(\ref{equation_definition_lemme_diagonalisation})$. Then, $\tr(V_{\lambda, n}^*V_{\lambda, n})$ is uniformly bounded and we have convergence of $(V_{\lambda, n})$ and $(U_{\lambda, n})$ in the following sense: for all $\psi \in \Hcal$,
    \begin{equation}
        V_{\lambda, n}\psi \to V_{\lambda}\psi,~~~~ V_{\lambda, n}^*\psi \to V_{\lambda}\psi,~~~~ U_{\lambda, n}\psi \to U_{\lambda}\psi~~~~\text{and}~~~~ U_{\lambda, n}^*\psi \to U_{\lambda}^*\psi.
    \end{equation}
\end{lemma}
\begin{proof}
    We recall that $2V_{\lambda, n} = c\big(\omega^{1/2} \xi_{\lambda, n}^{-1/2} - \omega^{-1/2}\xi_{\lambda, n}^{1/2}\big)J^*$ and $2U_{\lambda, n} = \big(\omega^{1/2} \xi_{\lambda, n}^{-1/2} + \omega^{-1/2}\xi_{\lambda, n}^{1/2}\big)$.\\

    \noindent First, we want to prove that $V_{\lambda, n}$ is uniformly bounded in Hilbert-Schmidt norm. Note that this is equivalent to showing that $V_{\lambda, n}^*$ is uniformly bounded in Hilbert-Schmidt norm. We recall that (see \eqref{equation_VVetoile})
    \begin{equation}
    	4V_{\lambda, n} V_{\lambda, n}^* = \omega^{-1/2} \xi_{\lambda,n} \omega^{-1/2} - 1 + \omega^{1/2} \xi_{\lambda,n}^{-1} \omega^{1/2} - 1.
    \end{equation}
    Then, 
    \begin{equation}
    	\tr \big(\omega^{-1/2}\xi_{\lambda}\omega^{-1/2} - 1 \big) = \frac{8\lambda}{\pi}\int_0^{+\infty} \frac{t^2 \|R_{-t^2}(\omega^2) f_n\|^2}{1 + 4\lambda \langle f_n | \omega R_{-t^2 }(\omega^2) f\rangle}\d t \leq 2\lambda\|\omega^{-1/2}f_n\|^2 \leq 2\lambda\|\omega^{-1/2}f\|^2
    \end{equation}
    and
    \begin{equation}
    	\tr \big(1 - \omega^{1/2} \xi_{\lambda}^{-1} \omega^{1/2}\big) \leq \frac{8\lambda}{\pi} \int_0^{+\infty} \|R_{-t^2}(\omega^2)\omega f_n\|^2 \d t = 2\lambda \|\omega^{-1/2}f_n\|^2 \leq 2\lambda \|\omega^{-1/2}f\|^2.
    \end{equation}
    Therefore,
    \begin{equation}
    	\|V_{\lambda, n}\|_{\mathrm{HS}}^2 = \|V_{\lambda, n}^*\|_{\mathrm{HS}}^2 = \tr\big(V_{\lambda, n}V_{\lambda, n}^*\big) \leq \lambda\|\omega^{-1/2}f\|^2,
    \end{equation}
    which proves that $V_{\lambda, n}$ is uniformly bounded in Hilbert-Schmidt norm.\\
    
    \noindent Now, we want to prove convergence of $U_{\lambda, n}$ and $V_{\lambda, n}$ (whose definition may be found in \eqref{equation_definition_lemme_diagonalisation}). Applying (\ref{equation_racine_quatrieme_perturbation}) and (\ref{equation_inverse_racine_quatrieme_perturbation}) of Lemma \ref{lemme_formules_integrales_racine_perturbation} we find
    \begin{equation}\label{equation_omega_xi}
        \omega^{-1/2} \xi_{\lambda, n}^{1/2} = 1 + \frac{8\sqrt{2}\lambda}{\pi}\int_0^{+\infty} \frac{t^4 |R_{-t^4}(\omega^2)f_n\rangle \langle R_{-t^4}(\omega^2) \omega^{1/2} f_n|}{1 + 4\lambda\langle f_n| R_{-t^4}(\omega^2)\omega f_n\rangle }\d t
    \end{equation}
    and 
    \begin{equation}\label{equation_xi_omega}
        \omega^{1/2}\xi_{\lambda,n}^{-1/2} = 1 - \frac{8\sqrt{2}\lambda}{3\pi}\int_0^{+\infty}\frac{|R_{-t^{4/3}}(\omega^2)\omega f_n\rangle \langle R_{-t^{4/3}}(\omega^2)\omega^{1/2} f_n|}{1 + 4\lambda\langle f_n | R_{-t^{4/3}}(\omega^2)\omega f_n\rangle}\d t.
    \end{equation}
     From (\ref{equation_omega_xi}) and (\ref{equation_xi_omega}), we can deduce that for all $\psi \in D(\omega)$, $\omega^{-1/2} \xi_{\lambda, n}^{1/2} \psi \to \omega^{-1/2}\xi_{\lambda}^{1/2} \psi$ and $\omega^{1/2} \xi_{\lambda,n}^{-1/2}\psi \to \omega^{1/2}\xi_{\lambda, n}^{-1/2}\psi$ by Lebesgue's dominated convergence theorem, and that $V_{\lambda, n}\psi \to V_{\lambda}\psi$ and $U_{\lambda, n}\psi \to U_{\lambda}\psi$.\\
    By density of $D(\omega)$ in $\Hcal$ it is enough to have uniform boundedness of $(V_{\lambda, n})$ and $(U_{\lambda, n})$ to have convergence on $\Hcal$. As $\tr(V_{\lambda, n}^*V_{\lambda, n})$ is bounded, $(V_{\lambda, n})$ is bounded, and this implies that $(U_{\lambda, n})$ is bounded too, because $U_{\lambda, n}^* U_{\lambda, n} = 1 + J^*V_{\lambda, n}^*V_{\lambda, n}J$, which concludes the proof.
\end{proof}
With those lemmas, we can finally prove Proposition \ref{proposition_convergence_diagonalisé}.

\begin{proof}[Proof of Proposition \ref{proposition_convergence_diagonalisé}] On the one hand, Lemma \ref{lemme_convergence_norme_seconde_quantification} and Lemma \ref{lemme_convergence_xil} show that $\dg(\xi_{\lambda, n})$ converges to $\dg(\xi_{\lambda})$ in the norm resolvent sense. On the other hand, Lemma \ref{lemme_convergence_transformations} and Lemma \ref{lemme_convergence_VU} show that we can choose $\mathbb{U}_{\lambda, n}$ such that $\mathbb{U}_{\lambda, n} \to \mathbb{U}_{\lambda}$ strongly. For $z \in \C \setminus (-\infty , -1]$, we have 
\begin{equation}
    R_z\big(\mathbb{U}_{\lambda, n} \dg(\xi_{\lambda, n})\mathbb{U}_{\lambda, n}^*\big) = \mathbb{U}_{\lambda, n} R_z\big(\mathrm{d}\Gamma(\xi_{\lambda, n})\big)\mathbb{U}_{\lambda, n}^*.
\end{equation}
As $\mathbb{U}_{\lambda, n}$ and $R_z\big(\mathrm{d}\Gamma(\xi_{\lambda, n})\big)$ are uniformly bounded, 

\begin{equation}
    \mathbb{U}_{\lambda, n} R_z\big(\mathrm{d}\Gamma(\xi_{\lambda, n})\big)\mathbb{U}_{\lambda, n}^* \to \mathbb{U}_{\lambda} R_z\big(\mathrm{d}\Gamma(\xi_{\lambda})\big)\mathbb{U}_{\lambda}^*
\end{equation}
strongly.
\end{proof}

With Propositions \ref{proposition_existence_trace}, \ref{proposition_existence_transformations_bogolioubov} and \ref{proposition_convergence_diagonalisé}, we have proved Theorem \ref{theorem1}.

\section{Proof of Theorem \ref{theorem2}}\label{section_preuve_theoreme2}

We divide the proof of the case $\omega^{-1/2}f\notin \Hcal$ in two parts; first, we discuss the existence of Bogoliubov transformations, and then we show the convergence of the diagonalized regularized model.

\subsection{Existence of Bogoliubov transformations}
In this section, we prove that, in general, there does not need to exist a Bogoliubov transformation associated with $\Vcal_{\lambda}$ defined by (\ref{equation_definition_lemme_diagonalisation2}) when $\omega^{-1/2}f\notin \Hcal$. Namely, if $\omega^{-1}f\notin \Hcal$, the Bogoliubov transformation never exists, and we give an example of an $\omega$ and a class of $f \in D(\omega^{-1})$ such that the Bogoliubov transformation does not exist.
\begin{proposition}[Non-existence of Bogoliubov transformations] \label{proposition_nonexistence_transformations}~ 
\begin{itemize}
    \item Assume $\omega^{-3/2}f\in \Hcal$ and $\omega^{-1}f \notin \Hcal$. We define $\xi_{\lambda}$ by $(\ref{equation_resolvante_renormalisee})$, $U_{\lambda}$ and $V_{\lambda}$ by $(\ref{equation_definition_lemme_diagonalisation})$, and $\Vcal_{\lambda}$ by $(\ref{equation_definition_lemme_diagonalisation2})$. Then, $\tr(V_{\lambda}V_{\lambda}^*) = +\infty$, therefore there does not exist a Bogoliubov transformation associated to $\Vcal_{\lambda}$.
    \item Assume $X = [1, + \infty)$, $\omega(k) = k$ and $\d\mu(k) = \d k$. Let $s\in (1/2, 1]$, $\varepsilon \in (0, s-1/2)$. There exists $f\in D(\omega^{-s})$ such that $\omega^{-s + \varepsilon}f \notin \Hcal$ which satisfies $\tr(V_{\lambda}V_{\lambda}^*) = +\infty$, for $\xi_{\lambda}$ defined by $(\ref{equation_resolvante_renormalisee})$, $U_{\lambda}$ and $V_{\lambda}$ by $(\ref{equation_definition_lemme_diagonalisation})$, and $\Vcal_{\lambda}$ by $(\ref{equation_definition_lemme_diagonalisation2})$.
\end{itemize}
\end{proposition}

\begin{proof}
    Let $f\in D(\omega^{-3/2})$. We recall that for $z\in \C \setminus \R_+^*$,
    \begin{equation}
        R_z(\xi_{\lambda}^2) = R_z(\omega^2) - 4\lambda \frac{\omega^{1/2}R_z(\omega^2) |f\rangle\langle f|R_z(\omega^2) \omega^{1/2}}{1 + 4\lambda z \langle f | \omega^{-1} R_z(\omega^2)f\rangle},
    \end{equation}
    and hence
    \begin{equation}\label{equation_nouvelle_definition_xilambda}
        \xi_{\lambda} = \omega + \frac{8\lambda}{\pi}\int_0^{+\infty} t^2 \frac{\omega^{1/2}R_{-t^2}(\omega^2) |f\rangle\langle f|R_{-t^2}(\omega^2) \omega^{1/2}}{1 - 4\lambda t^2 \langle f | \omega^{-1} R_{-t^2}(\omega^2)f\rangle}\d t
    \end{equation}
    and 
    \begin{equation}\label{equation_nouvelle_definition_xilambda_inverse}
        \xi_{\lambda}^{-1} = \omega^{-1} - \frac{8\lambda}{\pi}\int_0^{+\infty} \frac{\omega^{1/2}R_{-t^2}(\omega^2) |f\rangle\langle f|R_{-t^2}(\omega^2) \omega^{1/2}}{1 - 4\lambda t^2 \langle f | \omega^{-1} R_{-t^2}(\omega^2)f\rangle}\d t.
    \end{equation}
Note that (\ref{equation_nouvelle_definition_xilambda}) and (\ref{equation_nouvelle_definition_xilambda_inverse}) are different from (\ref{equation_valeur_integrale_xilambda}) and (\ref{equation_valeur_integrale_xilambda_inverse}) because we use a different definition of $\xi_{\lambda}$. Recall that $\lambda \leq 0$; we have 
    \begin{equation}
        V_{\lambda} V_{\lambda}^* = \omega^{-1/2}\xi_{\lambda} \omega^{-1/2} - 1 + \omega^{1/2}\xi_{\lambda}^{-1}\omega^{1/2} - 1,
    \end{equation}
    with
    \begin{equation}\label{xi_décoré}
        \omega^{-1/2}\xi_{\lambda} \omega^{-1/2} - 1 = \frac{- 8 |\lambda|}{\pi}\int_0^{+\infty} t^2\frac{R_{-t^2}(\omega^2)|f\rangle \langle f| R_{-t^2}(\omega^2)}{1 + 4|\lambda|t^2\langle f| \omega^{-1}R_{-t^2}(\omega^2)f\rangle}\d t
    \end{equation}
    and 
    \begin{equation}\label{inverse_xi_décoré}
        \omega^{1/2} \xi_{\lambda}^{-1} \omega^{1/2} - 1 = \frac{8|\lambda|}{\pi} \int_0^{+\infty}\frac{\omega R_{-t^2}(\omega^2)|f\rangle \langle f | R_{-t^2}(\omega^2)\omega}{1 + 4|\lambda|t^2\langle f |\omega^{-1}R_{-t^2}(\omega^2) f\rangle}\d t.
    \end{equation}
    We set 
    \begin{equation}
        A(t) = \frac{\omega R_{-t^2}(\omega^2)|f\rangle \langle f | R_{-t^2}(\omega^2)\omega - t^2 R_{-t^2}(\omega^2)|f\rangle \langle f| R_{-t^2}(\omega^2)}{1 + 4|\lambda|t^2\langle f| \omega^{-1}R_{-t^2}(\omega^2)f\rangle},
    \end{equation}
    which gives 

    \begin{equation}
        V_{\lambda}V_{\lambda}^* = \frac{8|\lambda|}{\pi}\int_0^{+\infty}A(t)\d t.
    \end{equation}
    We will now treat the two cases separately.\\

        $\bullet$ Assume that $\omega^{-1}f\notin \Hcal$. For a fixed $t \in \R_+$, we have 
        \begin{equation}
            \tr A(t) = +\infty
        \end{equation}
        as $\omega R_{-t^2}(\omega^2)f \notin \Hcal$ and $R_{-t^2}(\omega^2)f\in \Hcal$. Therefore, we have
        \begin{equation}
            \tr(V_{\lambda}V_{\lambda}^*) = + \infty.
        \end{equation}
    
        $\bullet$ Here we assume that $X = [1, + \infty)$, $\d \mu(k) = \d k$ and $\omega(k) = k$. Let $s\in (1/2, 1]$, $\varepsilon\in (0, s-1/2)$, $\alpha = s - \varepsilon -1/2\in (0, 1/2)$ and $f(k) = k^{\alpha}$. Then, $f\in D(\omega^s)$ and $\omega^{-s + \varepsilon}f\notin \Hcal$. In this particular case, we have
    \begin{equation}\label{estimation_dénominateur}
        \langle f | \omega^{-1}R_{-t^2}(\omega^2) f\rangle  = \int_1^{+\infty} \frac{k^{2\alpha}}{k(k^2 + t^2)}\d k= t^{2\alpha - 2}\int_{1/t}^{+\infty} \frac{p^{2\alpha}}{p(1 + p^2)}\d p = C_0 t^{2\alpha - 2}\big( 1 + O(t^{-2\alpha})\big).
    \end{equation}
    We set 
    \begin{equation}
        u(t) = \tr(A(t)) = \frac{\big\langle f \big| R_{-t^2}(\omega^2) (\omega^2 - t^2)R_{-t^2}(\omega^2)f\big\rangle}{1 + 4|\lambda|t^2\langle f|\omega^{-1}R_{-t^2}(\omega^2)f\rangle},
    \end{equation}
    that is 
    \begin{equation}
        u(t) = \big(1 + |\lambda|t^2 \langle f|\omega^{-1}R_{-t^2}(\omega^2)f\rangle\big)^{-1}\int_1^{+\infty}\frac{k^{2\alpha}(k^2 - t^2)}{(k^2 +t^2)^2}\d k.
    \end{equation}
    Now we decompose $V_{\lambda}V_{\lambda}^*$ into two parts; we choose $T > 0$ such that 
        \begin{equation}
            \forall t\geq T, \int_{0}^{1/t} \frac{p^{2\alpha}}{p(1 + p^2)}\d p \leq Ct^{-2\alpha}
        \end{equation}
        and we write
        \begin{equation}
            I = \int_0^{+\infty} \frac{p^{2\alpha}}{p(1+p^2)}\d p.
        \end{equation}
        For $t \geq T$, we have
        \begin{equation}
            |\langle f | \omega^{-1}R_{-t^2}(\omega^2) f\rangle - It^{2\alpha - 2}| \leq Ct^{-2}.
        \end{equation}
        Let $\tau \geq T$, we have
        \begin{equation}
            \left|\int_T^{\tau} u(t) \d t - \frac{1}{4|\lambda| I}\int_T^{\tau}\int_1^{+\infty} \frac{k^{2\alpha} t^{-2\alpha}(k^2 - t^2)}{(k^2 + t^2)^2}\d k \d t \right| \leq \int_T^{\tau}\int_1^{+\infty} \frac{k^{2\alpha} t^{-4\alpha}(k^2 + t^2)}{(k^2 + t^2)^2}\d k \d t.
        \end{equation}
        It is quite clear that the right-hand side is bounded as a function of $\tau$, therefore the integral $\int_T^{+\infty}u(t)\d t$ converges if and only if the integral $\int_T^{+\infty}\int_1^{+\infty} \frac{k^{2\alpha} t^{-2\alpha}(k^2 - t^2)}{(k^2 + t^2)^2}\d k \d t $ converges. We have
        \begin{equation}
            \mathcal{I}(\tau) := \int_T^{\tau}\int_1^{+\infty} \frac{k^{2\alpha} t^{-2\alpha}(k^2 - t^2)}{(k^2 + t^2)^2}\d k \d t = \int_{T}^{\tau} t^{-1}\int_{1/t}^{+\infty} \frac{p^{2\alpha}(p^2 - 1)}{(p^2 + 1)^2}\d p\d t.
        \end{equation}
        When $p \to + \infty$, 
        \begin{equation} 
        	\frac{p^{2\alpha}(p^2 - 1)}{(p^2 + 1)^2}\sim p^{2 - 2\alpha}
       	\end{equation}
   		and when $p \to 0^+$, 
   		\begin{equation} 
   			\frac{p^{2\alpha}(p^2 - 1)}{(p^2 + 1)^2} \sim - p^{2\alpha}. 
   		\end{equation}
   		As $\alpha < 1/2$, $\frac{p^{2\alpha}(p^2 - 1)}{(p^2 + 1)^2}$ is integrable. Let us show that its integral is positive. First, with $q = p^{-1}$,
        \begin{equation}
            \int_0^1 \frac{p^{2\alpha} (p^2 - 1)}{(p^2 + 1)^2}\d p = \int_1^{+\infty} \frac{q^{-2\alpha} (1 - q^2)}{(q^2 + 1)^2}\d q.
        \end{equation}
        This implies
        \begin{equation}
            I_0 := \int_0^{+\infty} \frac{p^{2\alpha} (p^2 - 1)}{(p^2 + 1)^2}\d p = \int_1^{+\infty} \frac{(p^{2\alpha}-p^{-2\alpha}) (p^2 - 1)}{(p^2 + 1)^2}\d p > 0.
        \end{equation}
        Therefore,
        \begin{equation}
            \mathcal{I}(\tau) \sim I_0\int_T^{\tau}\frac{\d t}{t} \sim I_0 \ln(\tau) \to + \infty, 
        \end{equation}
        which implies
        \begin{equation}
            \int_T^{+\infty}u(t)\d t = + \infty.
        \end{equation}       
        In order to prove that $\tr(V_{\lambda}V_{\lambda}^*)=+\infty$, it only remains to show that 
        \begin{equation}
            \left|\int_0^Tu(t)\d t \right|<+\infty.
        \end{equation}
         Indeed, we have
        \begin{align}
            \left|\int_0^T u(t) \d t\right| &\leq \int_0^T\int_1^{+\infty}\frac{k^{2\alpha}(k^2 + t^2)}{(k^2 + t^2)^2}\d k \d t \notag\\
            &\leq \int_0^T\int_1^{+\infty}k^{2\alpha - 2}\d k \d t \leq T\int_1^{+\infty}k^{2\alpha-2}\d k < +\infty.
        \end{align}
\end{proof}

\subsection{Convergence of the regularized Hamiltonian}

Here, we show that the diagonalized regularized Hamiltonian converges. Note that as Bogoliubov transformations don't exist in the limit, we will only prove convergence of the diagonalized regularized Hamiltonian. \\

\noindent We recall that $\lambda_n^{-1} = \lambda^{-1} - 4\langle f_n|\omega^{-1}f_n\rangle$. It is clear that $\lambda_n \to 0^-$ when $n\to + \infty$.

\begin{proposition}[Convergence of the regularized Hamiltonian] \label{proposition_convergence_diagonalisé2}
    Assume $\omega^{-1/2}f\notin \Hcal$. We have the following convergence 
    \begin{equation}
        \dg(\xi_{\lambda_n, n}) \to \dg(\xi_{\lambda})
    \end{equation}
    in the norm resolvent sense.
\end{proposition}

\noindent To prove Proposition \ref{proposition_convergence_diagonalisé2}, we use the following lemma.

\begin{lemma}[Convergence of $\xi_{\lambda_n, n}$]\label{lemme_convergence_xil2}
    The operators $\xi_{\lambda_n, n}$ converge to $\xi_{\lambda}$ in the norm resolvent sense, i.e.
    \begin{equation}
        \forall z \in \C \setminus R,~ \|R_z(\xi_{\lambda_n, n}) - R_z(\xi_{\lambda})\| \to 0.
    \end{equation}
\end{lemma}

\begin{proof}
    First, we prove that $\xi_{\lambda_n, n}^2$ converges to $\xi_{\lambda}^2$ in the norm resolvent sense. Let us recall that $R_z(\xi_{\lambda}^2) = \tilde{R}_z(\lambda, f)$ and $R_z(\xi_{\lambda_n, n}^2) = \tilde{R}_z(\lambda, f_n)$ with
        \begin{equation}
            \tilde{R}_z(\lambda, f) = \big(\omega^2 - z\big)^{-1} -\frac{4\lambda}{1 + 4\lambda z \langle f |\omega^{-1}(\omega^2 - z)^{-1} f\rangle}\omega^{1/2}(\omega^2 - z)^{-1}|f\rangle\langle f| \omega^{1/2}(\omega^2 - z)^{-1}.
        \end{equation}
        For a fixed $z\in \C \setminus \R$, we set 
        \begin{equation}
            B_{n} = \frac{4\lambda}{1 + 4 \lambda z \langle f_n |\omega^{-1} (\omega^2 - z)^{-1} f_n\rangle}~~~~\text{and}~~~~ \psi_n = \omega^{1/2}(\omega^2 - z)^{-1}f_n \in \Hcal.
        \end{equation}
        Then, for $\varphi \in \Hcal$,
        \begin{equation}
            \big\|\big(\tilde{R}_z(\lambda, f_n) - \tilde{R}_z(\lambda, f)\big)\varphi\big\| \leq \Big(\|\psi\|^2|B_n - B| + |B_n|\|\psi\|\|\psi_n - \psi\| + |B_n|\|\psi_n\|\|\psi_n - \psi\|\Big)\|\varphi\|.
        \end{equation}
        As $\omega^{-3/2}f_n \to \omega^{-3/2}f$, we have $B_n \to B$ and $\psi_n \to \psi$. This implies that
        \begin{equation}
            \|\psi\|^2|B_n - B| + |B_n|\|\psi\|\|\psi_n - \psi\| + |B_n|\|\psi_n\|\|\psi_n - \psi\| \to 0,
        \end{equation}
        and therefore $R_{z}(\xi_{\lambda_n, n}^2) \to R_z(\xi_{\lambda}^2)$ in the sense of the norm.
    We have proved in Lemma \ref{lemme_convergence_xil} that the norm resolvent convergence of $\xi_{\lambda_n, n}^2$ to $\xi_{\lambda}^2$ implies the norm resolvent convergence of $\xi_{\lambda_n, n}$ to $\xi_{\lambda}$, which concludes the proof.
\end{proof}

\noindent Proposition \ref{proposition_convergence_diagonalisé2} is now a direct consequence of Lemma \ref{lemme_convergence_xil2} and Lemma \ref{lemme_convergence_norme_seconde_quantification}.

\appendix
\section{Appendix}

In this appendix, we prove Lemma \ref{lemme_formules_integrales_racine_perturbation} that we recall here.

\begin{lemma*}
    Let $A$ be a self-adjoint positive operator, $\psi \in \Hcal$ and $\alpha \in \R$ such that $A+\alpha|\psi\rangle\langle \psi|$ is positive (this is satisfied if $\alpha \in \R_+$). We denote $T_{\alpha} = A + \alpha |\psi\rangle\langle \psi|$, then
    \begin{equation}\tag{\ref{equation_racine_carree_perturbation}}
        T_{\alpha}^{1/2} = A^{1/2} + \frac{2\alpha}{\pi}\int_0^{+\infty}\frac{t^2 |R_{-t^2}(A)\psi\rangle\langle R_{-t^2}(A) \psi|}{1 + \alpha \langle \psi | R_{-t^2}(A) \psi\rangle}\d t,
    \end{equation}

    \begin{equation}\tag{\ref{equation_inverse_racine_carree_perturbation}}
        T_{\alpha}^{-1/2} = A^{-1/2} - \frac{2\alpha}{\pi}\int_0^{+\infty} \frac{|R_{-t^2}(A)\psi\rangle\langle R_{-t^2}(A) \psi|}{1 + \alpha \langle \psi | R_{-t^2}(A) \psi\rangle}\d t,
    \end{equation}

    \begin{equation}\tag{\ref{equation_racine_quatrieme_perturbation}}
        T_{\alpha}^{1/4} = A^{1/4} + \frac{2\sqrt{2}\alpha}{\pi}\int_0^{+\infty}\frac{t^4 |R_{-t^4}(A)\psi\rangle \langle R_{-t^4}(A)\psi|}{1 + \alpha \langle \psi | R_{-t^4}(A)\psi\rangle}\d t,
    \end{equation}
    and
    \begin{equation}\tag{\ref{equation_inverse_racine_quatrieme_perturbation}}
        T_{\alpha}^{-1/4} = A^{-1/4} - \frac{2\sqrt{2}\alpha}{3\pi}\int_0^{+\infty}\frac{ |R_{-t^{4/3}}(A)\psi\rangle \langle R_{-t^{4/3}}(A)\psi|}{1 + \alpha \langle \psi | R_{-t^{4/3}}(A)\psi\rangle}\d t.
    \end{equation}
    where $R_{z}(A) = (A - z)^{-1}$ is the resolvent of $A$. In particular, if $A \geq 1$, this implies that $T_{\alpha}^{1/2} - A^{1/2}$ is trace-class and that 
    \begin{equation}\tag{\ref{equation_trace_racine_carree_perturbation}}
        \tr\big(T_{\alpha}^{1/2} - A^{1/2}\big) = \frac{2\alpha}{\pi} \int_0^{+\infty} \frac{t^2\big\| R_{-t^2}(A)\psi\big\|^2}{1 + \alpha\big\langle \psi | R_{-t^2}(A)\psi\big\rangle} \d t.
    \end{equation}
\end{lemma*}

\begin{proof}~\\
    1. For a positive real number $x$, we have 
    \begin{equation}
        \int_0^{+\infty} \frac{x}{x + t^2}\d t = \sqrt{x}\int_0^{+\infty} \frac{\d s}{1 + s^2} = \frac{\pi}{2}\sqrt{x}.
    \end{equation}
    As $A$ is a self-adjoint positive operator, we can assume that it is a multiplication operator : $(A \psi)(x) = a(x) \psi(x)$ a.e for $f\in \Hcal$. Then, for $\psi, \varphi \in \Hcal$,
    \begin{multline}
        \langle \psi | A^{1/2} \varphi\rangle = \int_X \overline{\psi(x)}a(x)^{1/2}\varphi(x)\d \mu(x) = \int_X \int_0^{+\infty} \overline{\psi(x)} \frac{2}{\pi}\frac{a(x)}{a(x) + t^2} \varphi(x) \d t \d \mu(x) \\ = \frac{2}{\pi} \int_0^{+\infty} \langle \psi | A(A + t^2)^{-1} \varphi\rangle \d t,
    \end{multline}
    and thus
    \begin{equation}
        A^{1/2} = \frac{2}{\pi} \int_0^{+\infty} A R_{-t^2}(A) \d t = \frac{2}{\pi}\int_0^{+\infty}\big(1 - t^2 R_{-t^2}(A)\big)\d t.
    \end{equation}
    As $T_{\alpha}$ is also a self-adjoint positive operator, we have
    \begin{equation}
        T_{\alpha}^{1/2} = \frac{2}{\pi} \int_0^{+\infty} \big(1 - t^2 R_{-t^2}(T_{\alpha})\big)\d t.
    \end{equation}
    Using Lemma \ref{lemme_resolvante_rang_un} we find
    \begin{equation}
        R_{-t^2}(T_{\alpha}) = R_{-t^2}(A) - \frac{\alpha}{1 + \alpha \langle \psi|R_{-t^2}(A)\psi\rangle}R_{-t^2}(A)|\psi\rangle\langle \psi| R_{-t^2}(A)
    \end{equation}
    and hence
    \begin{equation}
        T_{\alpha}^{1/2} = A^{1/2} + \frac{2}{\pi}\int_0^{+\infty} \alpha\frac{R_{-t^2}(A)|\psi\rangle \langle \psi| R_{-t^2}(A)}{1 + \alpha\langle \psi | R_{-t^2}(A) \psi\rangle}\d t.
    \end{equation}
    2. For any positive real number $x$, we have 
    \begin{equation}
        \int_0^{+\infty} \frac{\d t}{x + t^2} = x^{-1/2}\int_0^{+\infty} \frac{\d s}{1 + s^2} = \frac{\pi}{2}x^{-1/2},
    \end{equation}
    therefore
    \begin{equation}
        A^{-1/2} = \frac{2}{\pi}\int_0^{+\infty} R_{-t^2}(A) \d t,
    \end{equation}
    which together with Lemma \ref{lemme_resolvante_rang_un} gives (\ref{equation_inverse_racine_carree_perturbation}).\\ 
    
   \noindent 3. For any positive real number $x$, we have 
    \begin{equation}
        \int_0^{+\infty} \frac{x}{x + t^4}\d t = x^{1/4} \int_0^{+\infty} \frac{\d s}{1 + s^4} = \frac{\pi}{2\sqrt{2}}x^{1/4},
    \end{equation}
    therefore
    \begin{equation}
        A^{1/4} = \frac{2\sqrt{2}}{\pi}\int_0^{+\infty} A R_{-t^4}(A) \d t = \frac{2 \sqrt{2}}{\pi}\int_0^{+\infty} \big(1 - t^4 R_{-t^4}(A)\big)\d t,
    \end{equation}
    which together with Lemma \ref{lemme_resolvante_rang_un} gives (\ref{equation_racine_quatrieme_perturbation}).\\
    
 \noindent   4. For any positive real number $x$, we have
    \begin{equation}
        \int_0^{+\infty}\frac{\d t}{x + t^{4/3}} = x^{-1/4} \int_0^{+\infty} \frac{\d s}{1 + s^{4/3}} = \frac{3\pi}{2\sqrt{2}} x^{-1/4},
    \end{equation}
    therefore 
    \begin{equation}
        A^{-1/4} = \frac{2\sqrt{2}}{3\pi}\int_0^{+\infty} R_{-t^{4/3}}(A) \d t,
    \end{equation}
    which together with Lemma \ref{lemme_resolvante_rang_un} gives (\ref{equation_inverse_racine_quatrieme_perturbation}).\\

   \noindent 5. From (\ref{equation_racine_carree_perturbation}), $T_{\alpha}^{1/2} - A^{1/2}$ is trace class if and only if 
    \begin{equation}
        \int_0^{+\infty} \frac{t^2\big\| R_{-t^2}(A)\psi\big\|^2}{1 + \alpha\big\langle \psi | R_{-t^2}(A)\psi\big\rangle} \d t < +\infty,
    \end{equation}
    and this inequality is satisfied as $R_{-t^2}(A) \leq \frac{1}{1 + t^2}$.

\end{proof}

\section*{Acknowledgements}
I would like to thank Jonas Lampart for the many discussions we had, for his explanations, for his advise concerning this manuscript and for his proofreading. I also thank Nicolas Rougerie for his proofreading.

\bibliographystyle{abbrv}

\end{document}